\documentclass[a4paper,11pt]{amsart}

\usepackage{color}
\usepackage{cite}
\usepackage{comment}
\usepackage{amssymb}
\usepackage{amsfonts}
\usepackage{amscd}
\usepackage{slashed}
\usepackage{pdflscape}
\usepackage{tikz}
\usepackage{xcolor}
\usepackage{enumerate}
\usepackage{multirow}
\usepackage{stmaryrd}
\usepackage{ytableau}
\usepackage{stackengine}
\usepackage{MnSymbol}

\newcommand{\filmstar}{\filledstar}

\usepackage{graphicx}
\newcommand{\cev}[1]{\reflectbox{\ensuremath{\vec{\reflectbox{\ensuremath{#1}}}}}}

\newcommand{\hh}{{\hspace{.3mm}}}

\newcommand{\pdot}{{\boldsymbol{\cdot}}}

\usepackage{bbold}
\usepackage[backref,pagebackref,linkcolor=blue]{hyperref}
\usepackage{mathrsfs}

\newtheorem{theorem}{Theorem}[section]
\newtheorem{lemma}[theorem]{Lemma}

\newtheorem{proposition}[theorem]{Proposition}
\newtheorem{corollary}[theorem]{Corollary}

\theoremstyle{definition}
\newtheorem{definition}[theorem]{Definition}
\newtheorem{example}[theorem]{Example}

\theoremstyle{remark}

\usepackage{graphicx} 
\usepackage{amsmath} 
\usepackage{amsfonts}
\usepackage{amssymb}

\newcommand{\be}{\begin{equation}}
	
	\newcommand{\ee}{\end{equation}}

\newcommand{\dangle}{\rangle}

\newcommand{\ba}{\begin{array}}
	
	\newcommand{\ea}{\end{array}}

\newcommand{\beq}{\begin{eqnarray}}
	
	\newcommand{\eeq}{\end{eqnarray}}

\newtheorem{lm}{lemma}

\newtheorem{thee}{theorem}

\newtheorem{proo}{proposition}

\newtheorem{co}{corollary}

\newtheorem{rem}{remark}

\newtheorem{deff}{definition}

\newcommand{\bd}{\begin{deff}}
	
	\newcommand{\ed}{\end{deff}}

\newcommand{\bl}{\begin{lm}}
	
	\newcommand{\el}{\end{lm}}

\newcommand{\bp}{\begin{proo}}
	
	\newcommand{\ep}{\end{proo}}

\newcommand{\bt}{\begin{thee}}
	
	\newcommand{\et}{\end{thee}}

\newcommand{\bc}{\begin{co}}
	
	\newcommand{\ec}{\end{co}}

\newcommand{\brm}{\begin{rem}}
	
	\newcommand{\erm}{\end{rem}}

\usepackage{t1enc}
\def\frak{\mathfrak}

\newcommand{\newc}{\newcommand}

\renewcommand{\exp}{\operatorname{exp}}

\let\ccdot\cdot
\def\cdot{\hbox to 2.5pt{\hss$\ccdot$\hss}}

\newc{\aR}{\mbox{\boldmath{$ R$}}}
\newc{\aS}{\mbox{\boldmath{$ S$}}}
\newc{\aT}{\mbox{\boldmath{$ T$}}}
\newc{\aW}{\mbox{\boldmath{$ W$}}}

\newc{\aD}{\mbox{\boldmath{$ D$}}\hspace{-.2mm}}

\newc{\aK}{\mbox{\boldmath{$ K$}}}
\newc{\aL}{\mbox{\boldmath{$ L$}}}

\usepackage{amssymb}
\usepackage{amscd}

\newcommand{\Rho}{P}

\newcommand{\nn}[1]{(\ref{#1})}

\newcommand{\X}{\mbox{\boldmath{$ X$}}}
 
\newc{\obstrn}[2]{B^{#1}_{#2}}

\newcommand{\rpl}                         
{\mbox{$
		\begin{picture}(12.7,8)(-.5,-1)
			\put(0,0.2){$+$}
			\put(4.2,2.8){\oval(8,8)[r]}
		\end{picture}$}}

\newcommand{\lpl}                         
{\mbox{$
		\begin{picture}(12.7,8)(-.5,-1)
			\put(2,0.2){$+$}
			\put(6.2,2.8){\oval(8,8)[l]}
		\end{picture}$}}

\usepackage{ifthen}

\newc{\tensor}[1]{#1}
\newc{\Mvariable}[1]{\mbox{#1}}
\newc{\down}[1]{{}_{#1}}
\newc{\up}[1]{{}^{#1}}

\newc{\JulyStrut}{\rule{0mm}{6mm}}
\newc{\midtenPan}{\mbox{\sf S}}
\newc{\midten}{\mbox{\sf T}}
\newc{\midtenEi}{\mbox{\sf U}}
\newc{\ATen}{\mbox{\sf E}}
\newc{\BTen}{\mbox{\sf F}}
\newc{\CTen}{\mbox{\sf G}}

\def\sideremark#1{\ifvmode\leavevmode\fi\vadjust{\vbox to0pt{\vss
			\hbox to 0pt{\hskip\hsize\hskip1em
				\vbox{\hsize3cm\tiny\raggedright\pretolerance10000
					\noindent #1\hfill}\hss}\vbox to8pt{\vfil}\vss}}}%

\numberwithin{equation}{section}

\newcommand\sss{\scriptscriptstyle}

\renewcommand\X{{\mathscr X}}

\usepackage{shadow}
\begin{document}

	\renewcommand{\today}{}
	\title{
	Discrete Dynamics and Supergeometry
	}
	\author{Subhobrata Chatterjee${}^\diamondsuit$,   Andrew Waldron${}^\spadesuit$ \&
	Cem Yet\.{i}\c{s}m\.{i}\c{s}o\u{g}lu${}^\clubsuit$
	}
	
	\address{${}^{\diamondsuit, \spadesuit}$Department of Mathematics\\
		University of California\\
		Davis, CA95616, USA} \email{sbhchatterjee,wally@math.ucdavis.edu}
	
	\address{${}^{\clubsuit}$Department of Mathematics\\
		\.{I}stanbul Technical University\\
		\.{I}stanbul, 34469, Turkey} \email{yetismisoglu@itu.edu.tr}

	\vspace{10pt}

	\vspace{10pt}
	
	\renewcommand{\arraystretch}{1}

	\begin{abstract} 
	
We formulate a geometric measurement theory of dynamical classical systems possessing both continuous and discrete degrees of freedom.
The approach is covariant with respect to choices of clocks and canonically incorporates laboratories. The latter are embedded symplectic submanifolds
of an odd-dimensional symplectic structure.
When suitably defined,  symplectic geometry in odd dimensions is exactly the structure needed for covariance.
A fundamentally probabilistic
viewpoint allows classical  supergeometries to describe discrete dynamics. We solve the problem of how to construct probabilistic measures on supermanifolds given a (possibly odd dimensional)
supersymplectic structure.
This relies on a superanalog of the Hodge star for differential forms
and a description of probabilities by convex cones. We also show how stochastic processes such as Markov chains
can be described by supergeometry.

		\vspace{2cm}
	\end{abstract}

	\maketitle
	
	\pagestyle{myheadings} \markboth{Chatterjee, Waldron \& Yet\'{i}\c{s}m\'{i}\c{s}o\u{g}lu}{Measurement Theory}
	
	\tableofcontents
		
	\newpage
	
	\section{Introduction}
Many physical systems involve both continuous and discrete degrees of freedom. Standard classical mechanics typically treats continuous systems.
Our aim is to give a unified treatment of
dynamical systems with both discrete and continuous degrees of freedom.
We focus on models whose dynamics is locally determined. The key to our approach is to view measurement as fundamentally probabilistic.
Dynamics is often used to describe time evolution, so we also ensure that our approach is covariant with respect to choices of clocks, {\it \`a la} Einstein's theory of general relativity. This requires a a natural generalization of symplectic geometry to odd dimensional manifolds.

\smallskip

Just as symplectic geometry describes continuous classical mechanical systems, we claim that symplectic supergeometry describes dynamical systems with continuous and discrete degrees of freedom. This neccesitates a probabilistic approach.

\smallskip

In quantum physics,  action principles
such as
$$
S=\int \big(\lambda_i(z,\theta) \hh dz^i + \ell_a(z,\theta)\hh d\theta^a\big) \, ,
$$
where $z$ and $\theta$ are respectively ``Bose''
and ``Fermi'' coordinates, 
are employed to describe the quantum mechanics of particles with intrinsic spin (see for example~\cite{brink1976local,casalbuoni1976classical,Berezin:1976eg,Galvao:1980cu,bastianelli2006path}). Almost fifty years ago, Berezin and Marinov attempted  to ascribe a
physical interpretation to a classical 
theory described by such an action, but could not formulate a good notion of positive definite probability distributions~\cite{Berezin:1976eg} (see~\cite{Barducci} for another early attempt). We shall solve this problem.

\smallskip

Our treatment begins with the standard formulation of measurement theory in classical mechanics, based on ideas dating back to Boltzmann, Maxwell and Gibbs and later formalized by Koopman and von Neumann~\cite{koopman1931hamiltonian, neumann1932operatorenmethode}.  In this probabilistic approach
 the state of a system is described in terms of probability densities which  are dual to observables corresponding to the quantities being measured. More general than the notion of a probability density is a ``probability cone'' describing the space of states as a suitable subset of the dual to the space of observables~\cite{MR0229429,MR0229430,MR0249064,MR0242429,MR0297221, MR0342092,schwarz2020geometric,MR3752196}. 
Expectations of observables are given by the natural pairing between elements of these two spaces. This picture is described in Section~\ref{sec2}.

At first sight, the probabilistic picture seems very different to introductory approaches to classical mechanics involving idealized particles and their trajectories 
 in some ``configuration space'' manifold $Z$. However, the manifold~$Z$ still plays an important {\it r\^ole}, but one is now interested in a bundle of probability cones over~$Z$. We outline this theory in Section~\ref{sec2}. Most of the ideas given there are not novel (see for example~\cite{MR2492178,MR2433906,MR2222127, arnol2013mathematical, ArnGiv,busch2016quantum}). A key new concept however  is the notion of a laboratory. This is required to describe classical systems independently of any particular choice of clocks \cite{herczeg2018contact}. Broadly speaking, a laboratory is an instantaneous snapshot 
of the system that is described by a hypersurface in a configuration space that includes directions corresponding to time evolution. 

In the probabilistic picture the notion of time evolution is fundamentally stochastic. However, we wish to achieve covariance with respect to choices of time.
Therefore, rather than working with an even dimensional symplectic manifold as the mathematical model for phase-space, and 
a Hamiltonian function whose Hamiltonian vector field generates dynamics in terms of parameterized phase-space paths, 
 a unified description of dynamics is captured by odd symplectic manifolds:  In odd dimensions, a symplectic manifold describes generalized positions, momenta and time and may thus be viewed as a ``phase-spacetime''. The required data in this setting is minimal, {\it videlicet} a maximal rank closed two form 
whose kernel defines unparameterized paths. Some authors include the data of a volume form~\cite{he2010odd, lin2013lefschetz} in the definition of an odd symplectic manifold. Cosymplectic structures add the data of a closed one-form which in turn defines a volume form~\cite{cappelletti2013survey}. 
We drop any requirement of a volume form in our treatment of odd symplectic geometry; see Section~\ref{sec3}. Phase-spacetime hypersurfaces naturally  serve as laboratories whenever they are symplectic in the standard sense. 
Evolution of probabilistic states is  governed by the Liouville equation.
In Section~\ref{sec4} we give the clock independent/time covariantized analog
of the Liouville equation.

In any stochastic approach to dynamics it is important to ensure that dynamics evolves probabilities into probabilities, or in other words preserves positivity. More generally, any evolution of states must be confined to the probability cone. There is a simple mechanism guaranteeing this property namely re-express states as a positive square of a ``state function'' with respect to a suitable product; see again Section~\ref{sec4}.
This mechanism is in fact employed in the Born interpretation of quantum mechanics, see~\cite{schwarz2020geometric}. This manuever 
is  critical to our supergeometric description of discrete-plus-continuous dynamical systems.

Discrete dynamical systems can be described using supergeometry. Accomplished supergeometers can safely skip the first half of the Section~\ref{sec6} exposition which handles standard notions such as sheaves of exterior algebras, superfunctions, supervectors, super differential forms and super integral calculus \cite{manin1997introduction, witten2012notes,Deligne}. A key notion is the correspondence between (split) supermanifolds and vector bundles provided by the structure theorems of Batchelor~\cite{batchelor1979structure} and Koszul~\cite{koszul1994connections}. 
We also review symplectic supergeometry~\cite{Schwarz1989, Rothstein1990,Schwarz1993,schwarz1994superanalogs} and introduce its odd dimensional counterpart. The latter is necessary for time covariance of discrete-plus-continuous dynamics.

To construct probability cones on symplectic supergeometries both a measure and product are required. The former is for pairing states and observables, and the product defines positive squares and in turn a cone.
The product is constructed via a super Moyal product \cite{curtright2013concise, woit2017quantum}  which implements a super-analog of the standard Hodge pairing on exterior algebras equipped with a metric.
This is described at the end of Section~\ref{sec6}, where our probabilistic treatment of symplectic supergeometry is also given.

Having developed a probability theory on symplectic supergeometries, it is not difficult to construct instantaneous (super) laboratories and study evolution via a super Liouville equation. This is presented in Section~\ref{sec7}. The paper concludes with a simple example.

\section{States, Laboratories and Measurements} \label{sec2}
	A measurement typically amounts to a set of numerical outputs characterizing the state of a system.
	These  ought include errors reflecting uncertainties  associated to measuring devices. Indeed, in practice one usually can only  hope to have a statistical description of a physical system. In that light,   states of a physical system should be viewed as probability distributions each of which 
	assigns a positive, normalizable, measure on some space of physically observable events. 


\smallskip
	
More generally, 	
	a set of positive measures 
may be treated as  a {\it  {convex} cone} ${\mathcal C}\,\slashed{\ni} \, 0$ in a topological vector space ${\mathcal V}$, meaning some subset of non-zero vectors that is  closed under addition and multiplication by positive scalars; we shall denote this by~${\mathcal C}<{\mathcal V}$.
When  ${\mathcal V}$ is equipped with a norm~$||\hh\pdot\hh||$, this can be used to  normalize states, and in that case we call $({\mathcal C}<{\mathcal V},||\hh\pdot\hh||)$ a {\it probability cone}.

An {\it observable} is any linear functional~$X:{\mathcal V} \to {\mathbb R}$ (say). The {\it expectation} of the observable $X$ in the {\it state} $\Psi\in{\mathcal C}$ is then 
$$
\langle X\rangle_\Psi = \frac{X(\Psi)}{||\Psi||}\, .
$$

\begin{example}\label{frogsRus}
Let 
${\mathcal V}$ be $\mathbb{R}^2$ equipped with the standard $L^1$-norm
and ${\mathcal C}$
 be the (closed) first quadrant $\mathbb{R}^2_+$ with the origin  $(0,0)$ removed. A normalized state $\Psi=(p_{\tt 0} , p_{\tt 1})
 \in \mathcal{C}$ 
 is thus  a pair of probabilities summing to unity.
The expectation of the observable~$X=(x_{\tt 0} \: x_{\tt 1})$
is
\begin{equation*}
		\langle X\rangle_\Psi
			=x_{\tt 0}p_{\tt 0}+x_{\tt 1}p_{\tt 1} \, .
\end{equation*}
In a probability context, the observable $X$  is called a random variable.
\hspace*{\fill}{$\blacksquare$}

\end{example}

\medskip

\noindent
The first quadrant in the above  
is an example of a convex polyhedral cone.
Its one-dimensional corners define vectors,
the expectations of    whose canonical duals (the covectors $\begin{pmatrix} 1\!&\!0\end{pmatrix}$ and $\begin{pmatrix} 0\!&
\!1\end{pmatrix}$)
define probabilities ($p_{\tt 0}$ and  $p_{\tt 1}$) in the standard sense. Later we will encounter more general probability cones (see in particular Example~\ref{conifer}).

\medskip
Classical mechanics is most often formulated as a theory of trajectories in some configuration/phase-space, the latter of which we will model by some manifold $Z$. As intimated earlier, from a probabilistic viewpoint 
one should supplant  trajectories by suitable functions of  $Z$
and somehow view these 
as either states or observables.
Hence we make the following definition for states of a classical dynamical system.

\begin{definition}\label{mylozengeisblack}
Let $C \to \mathcal{S}Z\to Z$ be a subbundle of some  vector bundle $V\to\mathcal{V}Z \to Z$ over a manifold $Z$, whose fibers are a cone~$C \subset V$ (and
$0\in C$). We shall call non-zero sections of the bundle ${\mathcal S}Z$ \textit{pre-states}. When  the section space $\Gamma({\mathcal V}Z)$ is equipped with a norm function, normalizable pre-states are termed {\it states}. The set of all states is then a normed space and is termed  
 the {\it state space}.\hfill$\blacklozenge$\end{definition}

%

In the above definition, we have included the zero vector in $C$ so that sections of~${\mathcal S}Z$ with compact support are allowed. The (normalizable) section space $\Gamma({\mathcal S}Z)$ plays the {\it r\^ole} of the probability cone~${\mathcal C}$ discussed above, while $\Gamma({\mathcal V}Z)$ corresponds to the topological vector space ${\mathcal V}$.
Supposing~$V$ is a finite dimensional vector space and that the  bundle of cones $\mathcal{S}Z\to Z$ is a subbundle of some vector bundle $V\to\mathcal{V}Z \to Z$, we then have a natural dual vector bundle $V^* \to \mathcal{V}^*\!Z \to Z$, and may 
  make the following definition.
	\begin{definition}
		Sections of the dual bundle $V^* \to \mathcal{V}^*\!Z \to Z$ are called \textit{preobservables}. If~$Z$ is equipped with a measure $\mu$, when a preobservable $\chi$ obeys
		\begin{equation*}\label{finite}
		\int_Z \mu\hh \chi(\Psi)<\infty
		\end{equation*}
for all states $\Psi$, then $\int_Z \mu \hh \chi(\pdot)$ is  termed  an {\it observable}.
	\hfill$\blacklozenge$\end{definition}
\noindent	Given the data of a measure, ``nice''
 preobservables determine  observables.
When such is available, we shall often  blur the distinction between observables and preobservables.

\begin{example}\label{evenexp}
Let $(M^{2n},\omega)$ be the pair of a compact, even-dimensional manifold and a non-degenerate $2$-form $\omega$. Then we may take $Z=M^{2n}$, $V={\mathbb R}$,  $C={\mathbb R}_{\geq 0}$ and ${
\mathcal S} Z= M\times {\mathbb R}_{\geq 0}$.
The state space  is the space of positive but not everywhere vanishing functions  $\Gamma({\mathcal S}Z)\backslash 0$
with  norm $||\hh\pdot\hh || = \int \omega^{\wedge n} |\hh\pdot\hh |$.
With respect to the measure $\mu := \omega^{\wedge n}$,  any  smooth function~$\chi$ of $M^{2n}$ is a preobservable
because we may define a linear functional on states~$\Psi$ by the product action
$$
\chi(\Psi)=\chi \Psi\, .
$$
Note that this relies on the algebra of functions on $M^{2n}$.
The expectation of the observable  $X=\int \mu \hh \chi(\pdot)$ is
$$
\langle X \dangle_{\Psi,\mu} = \frac{\int_{M^{2n}}\omega^{\wedge n} \chi\Psi   }{\int_{M^{2n}}\omega^{\wedge n}   \Psi   }\, .
$$
Blurring the distinction between preobservables and observables, we may view the latter as the set of smooth functions on $M^{2n}$.
\hspace*{\fill}{$\blacksquare$}
\end{example}


Observe that if ${\mathcal V}Z$ is a {\it normed bundle}, meaning that there is a smooth map 
$$
|\,\pdot\, | :\Gamma({\mathcal V}Z) \to C_+^\infty Z
$$	
taking values in positive functions such that $|\, \pdot\,|$ gives a fiberwise norm,
then 
given a measure on $Z$, we obtain the $L^1$-norm
$$
||\hh\pdot\hh || = \int_Z \mu\,  |\,\pdot\, |
$$
on $\Gamma({\mathcal V}Z)$. This mechanism was used in the above example. 
In some contexts a normed bundle is termed a Banach bundle.

As we are henceforth primarily concerned with probability cones constructed as section spaces of vector bundles  ${\mathcal V}Z$, there arise situations where measurements may not access all of $Z$. It may even be that only some submanifold $\Sigma\hookrightarrow Z$ is endowed with norms and measures. Along $\Sigma$ we naturally get a vector bundle ${\mathcal V}\Sigma:= {\mathcal V}Z|_\Sigma$ by restriction 
and in the same way a bundle of cones ${\mathcal S}\Sigma$.
This leads to the notion of a laboratory:

	\begin{definition}
		Any triple $(\Sigma,\mu,||\hh \pdot\hh||),$ where $\Sigma \hookrightarrow Z$ is any submanifold equipped with a measure $\mu$ and a norm $||\,\pdot\, ||$ on $\Gamma({\mathcal V} \Sigma)$, is called a \textit{laboratory}.
	\hfill$\blacklozenge$\end{definition}
\noindent	We can now define {\it laboratory observables} $X=\int_\Sigma \mu \chi(\pdot)$ in terms of  preobservables $\chi\in \Gamma({\mathcal V}^*\!Z)$ subject to  $\int_\Sigma \mu\hh \chi(\Psi)<\infty$
for all $\Psi\in \Gamma({\mathcal S}Z)$ that are normalizable along $\Sigma$.
Then the expectation of $X$ 
as measured by a laboratory $(\Sigma,\mu,||\, \pdot\, ||)$ with respect to some state $\Psi \in \Gamma({\mathcal S}Z)$ is given by
	\begin{equation*}
	\langle
	X\rangle_{\Psi,(\Sigma,\mu,||\, \pdot\, ||)}:= \frac{\int_{\Sigma}
	\mu \hh \chi(\Psi)}{ ||\Psi||}
\, .	\end{equation*}
%
%
%

\smallskip
In some cases laboratories may inherit measures and norms from the underlying manifold $Z$ when this has sufficiently much data, as is the case in Example~\ref{evenexp}.
\begin{example}
Continuing from Example~\ref{evenexp}, 
let $\Sigma^{2k}$ be any compact submanifold
along which the pullback $\omega^*$ of $\omega$ is non-degenerate (so that $(\Sigma,\omega^*{} ^{\wedge k})$ is itself symplectic).
Then $\big(\Sigma^{2k},\omega^{*\wedge k},\int_{\Sigma^{2k}} \omega^{*\wedge k} |\hh\pdot\hh |\big)$ is a laboratory.\hfill$\blacksquare$
\end{example}

The following is an example where $Z$ itself is not equipped with norms or measures.
	\begin{example}
	Carrying on from Example~	\ref{frogsRus} take $V={\mathbb R}^2\supset \mathbb{R}^2_+=C $
	and consider the manifold~$Z={\mathbb R}\ni t$.
Choosing laboratories $(\Sigma=\{t\},\mu=1,
||\, \pdot\, ||=\sum_{\{{\tt 0},{\tt 1}\}}|\, \pdot \, |
)$, the observable $X=\big(x_{\tt 0}(t)\: \: x_{\tt 1}(t) \big)$ 
has expectation
$$
\langle X\rangle_{\Psi,(\Sigma,\mu, ||\, \pdot\, ||)}
		=
		\frac{x_{\tt 0}(t)\pi_{\tt 0}(t)+x_{\tt 1}(t)\pi_{\tt 1}(t)}{\pi_{\tt 0}(t)+ \pi_{\tt 1}(t) }\, ,	$$
with respect to the state $\Psi=\big(\pi_{\tt 0}(t), \pi_{\tt 1}(t) \big)$ in these laboratories.
\hfill$\blacksquare$
	\end{example}
\noindent	
Viewing the coordinate $t$  as  time, the above example includes dynamics.
In this sense, $	(\pi_{\tt 0}(t), \pi_{\tt 1}(t) )$  is a ``dynamical state''.

\section{Dynamical Systems}
	

\label{sec3}

A dynamical system is the data of a configuration space  with  evolution. For continuous systems this is typically modeled by a smooth manifold equipped with a parameterized path through each point. For the case of particles moving in space, an important insight of Einstein was that parameterized paths could be replaced by unparameterized ones in a spacetime manifold. 
This perspective both allowed time to be treated covariantly and made manifest the geometry of space and time.
In classical mechanics, dynamics can be described by parameterized paths in  phase-space. But again, it is often propitious, for both reasons of covariance and geometry, to consider unparameterized paths in a phase-spacetime. In both cases, one is (at least locally)  dealing with a manifold foliated by paths. 
\smallskip

Rather than  working with a system of paths, it  is useful to describe physics modeled upon a manifold by local systems of equations. In very many situations the data of a suitable  two-form suffices for this.

	\begin{definition}\label{nothings}
		A \textit{symplectic manifold} is a pair $(Z,\omega)$ where $Z$ is a smooth manifold and $\omega$ is a maximally non-degenerate closed two-form.
		\hfill$\blacklozenge$\end{definition}
			\noindent
			We also need a notion of normalizability for functions on a symplectic manifold~$Z$.
When~$Z$ has even dimension $2n$, we say that\
$\Psi\in C^\infty Z$ is {\it normalizable} if $\int_Z \omega^{\wedge n}|\Psi|<\infty$.
For $Z$ of odd dimension $2n+1$, we call 
	$\Psi\in C^\infty Z$ is {\it normalizable} if $\int_\Sigma \omega^{*\wedge n}|\Psi|<\infty$ for all hypersurfaces $\Sigma\hookrightarrow Z$ along which the pullback $\omega^*$ is non-degenerate.	
	
	\medskip		
 A foliation by paths yields a line subbundle of the tangent bundle. At each point this is determined by the span of the tangent vector to the path determined by the foliation. 	
Thus to  describe dynamics, we would like to assign a distinguished line to each tangent space of  $(Z,\omega)$. 	\begin{definition}
		A {\it dynamical phase-space} is a pair, denoted $(Z,\omega,L)$,  given by a symplectic manifold 
$(Z,\omega)$ and a line subbundle $L \subset TZ$ subject to the compatibility condition that in a neighborhood of every point in $Z$, there is some 
non-vanishing $\rho\in \Gamma(L)$ such that 
$$
\mathcal{L}_{\rho} \omega =0\, .$$	
When $L$ is equipped with a global ray subbundle $R\subset L$,  the dynamical phase-space $(Z,\omega,L)$ is termed  {\it directed}.
	\hfill$\blacklozenge$\end{definition}
An even-dimensional  symplectic manifold equipped with a suitably generic Hamiltonian determines, at least locally, a dynamical phase-space through the tangent lines to  integral curves of the corresponding Hamiltonian vector field.	
	
	\smallskip
	
	Directed dynamical phase-spaces enjoy an everywhere forward notion of time evolution.  Moreover their corresponding line bundles $L$ are trivial, meaning~$L=M\times {\mathbb R}$. Often we will not be concerned with the distinction between the directed and non-directed cases, since our main aim is to establish local theories of measurement. 
	
\smallskip	
	
A vector $\rho$ such that ${\mathcal L}_{\rho}\omega =0$ for some line subbundle $L\subset TZ$ is called {\it symplectic}, and flows generated from such are termed \textit{symplectomorphisms}.
In Definition~\ref{nothings} we quite deliberately did not specify the dimension parity of $Z$.
When this is even, we relabel~$Z$ by~$M$ and call~$(M,\omega)$ a {\it phase-space} (oftentimes this is taken as the definition of a symplectic manifold), while the odd case is termed a {\it phase-spacetime}. The manifold~$M$ describes generalized positions and momenta, while $Z$ in addition covariantly incorporates time.

	\begin{example} \label{oddsympdynamics}
		Let  $H(p,q,t)$ 
		be a time dependent Hamiltonian on the phase-space $M=({\mathbb R}^2\ni(p,q),dp\wedge dq=:\bar \omega)$. This data can be rephrased as a phase-spacetime $Z=(M\times {\mathbb R}\ni t,\omega)$ by viewing $H$ as a function of $Z$ and setting
		$$
		\omega:=\Big(dp  +\frac{\partial H}{\partial q}dt\Big)\wedge \Big( dq -   \frac{\partial H}{\partial p}dt\Big)
\, .
		$$
The 	line bundle $L\subset TZ$ is given by the span of the vector field
$$
\rho = \frac{\partial}{\partial t}+\frac{\partial H}{\partial p} \frac{\partial}{\partial q}-\frac{\partial H}{\partial q} \frac{\partial}{\partial p}\, .
$$	
Notice that $\omega(\rho,\pdot)=0$, so clearly
 ${\mathcal L}_\rho \omega =0$.
Of course,  the above symplectic form $\omega$ is the pullback of the symplectic form $\Omega:=dp\wedge dq - dE \wedge dt$  on ${\mathbb R}^4\ni(p,q,E,t)$  
 to the  ``on-shell'' hypersurface $E=H(p,q,t)$.
 \hfill$\blacksquare$
	\end{example}	
 
 	In fact any phase-space $(M,\omega)$ can be treated as a section of the trivial line bundle $Z=M\times\mathbb{R}$ equipped with the symplectic form $\pi^*\omega$ where $\pi:Z\to M$ denotes the canonical projection. The pair $(Z,\pi^*\omega)$ is a phase-spacetime as is any contact manifold~\cite{herczeg2018contact,Corradini}. Not all phase-spacetimes can be obtained so easily.
	
\medskip	

Three symplectic manifolds appear  in Example~\ref{oddsympdynamics}:
$
({\mathbb R}^4, \Omega)$,
$({\mathbb R}^3, \omega)$
and $
({\mathbb R}^2, \bar\omega)
$.
To describe dynamics, the first requires the choice of some hypersurface ($E=H(p,q,t)$ say) along which evolution is described in terms of unparameterized paths. The second is in some sense optimal, since 
the symplectic form $\omega$ determines unparameterized paths describing evolution. The third 
describes dynamics in terms of parameterized paths for which a choice of Hamiltonian is required.
		Indeed, in odd dimensions, 
the data of a symplectic form is primal in the following sense.	
	\begin{proposition}\label{odd-dynamical}
		Every phase-spacetime 
		$(Z,\omega)$  canonically defines a dynamical phase-space. 
	\end{proposition}
	
	\begin{proof}
		Let $\rho$ denote any non-zero vector field solving
		\begin{equation}
			\iota_{\rho}\omega=0\, . 
			\label{froggy}
		\end{equation}
		Because $\omega$ is a maximal rank two-form, it defines a linear map $\Gamma(TZ)\to \Gamma(T^*Z)$ whose kernel is everywhere one-dimensional. By our definition, the vector $\rho$ 
need not be {\it everywhere} non-vanishing and therefore
such a vector can be constructed on any patch in $Z$. Next,   using  Cartan's magic formula, we compute
		\begin{equation*}
			\mathcal{L}_{\rho}\omega=d\iota_{\rho}\omega+\iota_{\rho}d\omega=0 \, .
		\end{equation*}
		The term $\iota_{\rho}d\omega$ vanishes in the above because the symplectic form is closed. Thus patch-wise, and hence everywhere,
 $\operatorname{span}(\rho)$ 	defines 
		 a line subbundle $L$ with the desired~$\mathcal{L}_{\rho} \omega =0$ property.
	\end{proof}

Our primary focus is now on phase-spacetimes. In particular we need to set up measurement theory and dynamics for states.
	
\section{Phase-spacetime Measurements}	
\label{sec4}
	
	As discussed earlier, the data of a phase-spacetime $(Z,\omega)$ encodes the dynamics of a system.
	Therefore one might expect that knowledge of the state of a system at one instant of ``generalized time'' should determine the state for all times. In other words, a state on a phase-spacetime should encode the information of the system's state at all times.
	
\begin{definition}
Let $(Z,\omega)$ be a phase-spacetime.
Then any normalizable function $0$~$<$~$\Psi\in C^\infty Z$
subject to 
$$
{\mathcal L}_\rho \Psi =0\, ,
$$
for all vector fields $\rho$ obeying $\omega(\rho,\pdot\hh)=0$, is termed a {\it dynamical state}.
\hfill$\blacklozenge$\end{definition}

%
%

	

	
%

Next we discuss how dynamical states can be measured. 
The symplectic form $\omega$ on a phase-spacetime does not define a volume form, so expectations of observables cannot be defined as was done in Example~\ref{evenexp}. This is a feature not a bug, since experiments are often viewed as taking place at   some instant in ``time''. 

\begin{definition}
Let $(Z,\omega)$ be a phase-spacetime.
Then any hypersurface $\Sigma$ such that the pullback of $\omega$ is non-degenerate is termed an {\it instantaneous laboratory}.
\hfill$\blacklozenge$\end{definition}

An instantaneous laboratory gives a generalized notion of  ``when'' one plans to measure a dynamical state. 
Observables encode what quantity is being measured. This ought be specified independently of any particular choice of laboratory. 
We note that one can also consider more general laboratories given by even dimensional submanifolds along which $\omega$ is non-degenerate. This is useful when one plans to measure only some subset of observable quantities, such as is the case in statistical mechanics.
The ideas presented below generalize straightforwardly to these cases. 

In a phase-spacetime context we will deliberately blur the distinction  between preobservables and observables, and hence make the follow definition.
\begin{definition}
Let $(Z^{2n+1},\omega)$ be a phase-spacetime. Then any function~$X\in C^\infty Z$
such that $\int_{\Sigma} \omega^{*\wedge n}\Psi X<\infty$ for all instantaneous laboratories $\Sigma$ and dynamical states $\Psi$,
 is called an {\it observable}.
\hfill$\blacklozenge$\end{definition}

The data of a dynamical state and an instantaneous laboratory suffices to compute expectations of observables at any ``generalized instant in time''.

\begin{definition}
Let $(Z^{2n+1},\omega)$ be a phase-spacetime, $\Psi$ be a dynamical state and $\Sigma$ an instantaneous laboratory. Then the {\it expectation of $X$ as measured by the laboratory $\Sigma$ with respect to the state $\Psi$} is defined to be
$$
\langle X\dangle_{\Psi,\Sigma}:=
			\frac{\int_{\Sigma} \omega^{*\wedge n} \, X\Psi }{\int_{\Sigma} \omega^{*\wedge n}\,  \Psi} \, .
$$
	
\hfill$\blacklozenge$\end{definition}

\begin{example}[Massless relativistic particle]
We can describe a photon moving in the spacetime
$$
ds^2 = -dt^2 + h_{ij}(x,t) dx^i dx^j
$$
as follows. For $(Z,\omega)$ we take $Z={\mathbb R}^7_\star\ni(0\neq p_i,x^i,t)$ and
$$
\omega =\Big(dp_i - p^k \Gamma_k{}^j{}_i p_j\hh
\frac{dt}E
\Big) \wedge
\Big( d x^i -  p^i\hh
\frac{dt}{E}
\Big)\, .
$$
In the above indices are raised and lowered with the spatial metric $h_{ij}$, $\Gamma_k{}^j{}_i$ are the corresponding Christoffel symbols, and the photon energy
$$
E:=\sqrt{p_i h^{ij} p_j}\, .
$$
For simplicity let us make the simple spatial metric ansatz
$$
h_{ij}(x,t) = f(t)^2\, \delta_{ij}\, ,
$$
which for monotonic increasing $f(t)$ describes an expanding universe.
Let us consider an instantaneous laboratory $\Sigma=\{(p_i,x^i,0)\}$.
The measure $\mu$ obtained by pulling back~$\omega$ is given  by
$$
\frac1{3!}\hh \mu = \bigwedge_{i=1..3}dp_i \wedge dx^i\, .
$$
Now suppose that the state of the system in the initial laboratory is given by some non-negative function $\psi(p,x)$ that has unit  $L^1$-norm  with respect to $\mu$.

Now, it is easy to check that the vector 
$$
\rho = f(t) \left( \frac{\partial}{\partial t}+ \hat p^{\hh i} \frac{\partial}{\partial x^i} \right) \, ,
$$
is symplectic,
where $\hat p$ is a unit vector with respect to the spatial Euclidean metric which we from now on use to raise and lower indices.
Because $f$ is monotonic, we may introduce a new variable~$\tau$ such that $d\tau/dt=f^{-1}$.
It is easy to solve ${\mathcal L}_\rho \Psi=0$ with this initial data, giving
$$
\Psi =\psi(p, x - \hat p \tau)\, .
$$
The expectation of an observable $X(p,x,\tau)$ with respect to a constant $\tau=T$ laboratory is then given by
$$
\langle X\rangle_{\Psi,\tau=T}=\int \big(\prod_{i=1..3} dp_i dx^i\big) X(p,x,T)
\psi(p, x - \hat p \hh T)\, .
$$
In the case where $\psi$ is a unit mass function strongly peaked around a point $(p,x)=(P,X)$, one has
$$
\langle p_i\rangle_{\Psi,\tau=T}\approx P_i \:\mbox{ and }
\langle x^i\rangle_{\Psi,\tau=T}\approx \hat P^i T\, .
$$
\hfill{$\blacksquare$}
\end{example}

Observe that in the above example, 
positivity of the initial state $\psi(p,x)$ implies the same for the dynamical state $\Psi$ everywhere in $Z$ so that indeed $\Psi\in \Gamma({\mathcal S}Z)$. This follows generally because the  dynamical state is obtained by Lie dragging the initial state along the direction defined by the line bundle $L$.

Now suppose we are given a one-parameter family of laboratories $\Sigma_\tau$ on $(Z^{2n+1},\omega)$. Because these are necessarily  transverse to the line bundle $L$, we may at least locally realize these as the exponential map of some section of $\rho\in \Gamma(L)$, {\it id est}
$$
\Sigma_\tau = \exp(\tau \rho)(\Sigma)\, ,
$$
where $\Sigma:=\Sigma_0$.
This yields an evolution law for expectations
$$
\frac{d( \langle X\dangle
_{\Psi,\Sigma_\tau})}{d\tau}\Big|_{\tau=0}=\frac{\int_\Sigma \omega^{*\wedge n} ({\mathcal L_\rho X) \Psi}}{\int_\Sigma \omega^{*\wedge n} \Psi}\, .
$$
The above result is easily established remembering that $\rho$ preserves both $\omega$ and $\psi$.


We next ask the silly question, why was the dynamical state $\Psi$ a positive function? The silly answer is that $\Psi$ must be a square
\begin{equation}\label{silly}
\Psi = \Phi^2\, ,
\end{equation}
since this is manifestly positive.
In more general models involving states described by  sections of more complicated vector bundles, this question and its answer are no longer so trivial, but can be addressed via generalized ``squares''.

\begin{definition}
Let $(Z,\omega)$ be a phase-spacetime. The {\it fundamental state space} is the set of smooth functions $\Phi\neq 0$ whose squares are normalizable, and that are subject to
\begin{equation}\label{Schrodinger}
{\mathcal L}_\rho \Phi =0
\end{equation}
for all vector fields $\rho$ obeying $\omega(\rho,\pdot\hh)=0$. Such a function  is termed a {\it state function}.
\hfill$\blacklozenge$\end{definition}

\begin{lemma}
The square of any state function is a dynamical state.
\end{lemma}

\begin{proof}
Let $\Phi$ be a state function.
We only need establish that ${\mathcal L}_{\rho} \Phi^2 =0$
for any vector field $\rho$ obeying $\omega(\rho,\pdot\hh)=0$. This is a direct consequence of the Leibniz rule for Lie derivatives.
\end{proof}

Finally, observe that when expressed in terms of state functions, expectations of observables take a form highly reminiscent of that for quantum mechanical systems:
$$
\langle X\dangle_{\Psi,\Sigma}:=
			\frac{\int_{\Sigma} \omega^{*\wedge n} \, \Phi X\Phi }{\int_{\Sigma} \omega^{*\wedge n}\,  \Phi^2} \, .
$$ 
State functions are in this sense classical analogs of quantum mechanical wave functions and Equation~\nn{Schrodinger} is the analog of the Schr\"odinger equation.

		\section{Super Geometry} \label{sec6}
	We are primarily interested in discrete dynamical systems associated to symplectic supermanifolds. 
The space of functions of $m$ Grassmann variables is a vector space of dimension $2^m$. This suggests studying independent coin flips or Bernoulli random variables. Indeed consider an $m$-bit computer. In a perfect world the state of this computer could be described by an $m$-element list of ones and zeros:
	$$
	\underbrace{{\tt 110010010000111111011}\cdots {\tt 001110001101}}_{m \ \rm terms}
	$$
	However in principle, even classical measurements of discrete systems are error-prone, so in general we only know a list of $2^m$ probabilities.
Moreover, observables are then random variables assigning a real number (say) to each of $2^m$ configurations of ${\tt 0}$'s and~${\tt 1}$'s. This data can encoded by a superfield:  	
	Introducing Grassmann variables~$\{\theta_{\tt 1},\ldots,\theta_{\tt m}\}$, the single superfield
	$$
	X(\theta) = x_{{\tt 000}\cdots{\tt 00}} + x_{{\tt 100}\cdots{\tt 00}} \theta_{\tt 1} + \cdots 
	+x_{{\tt 101100}\cdots{\tt 00}} \theta_{\tt 1} \theta_{\tt 3} \theta_{\tt 4} + \cdots
	+
	x_{{\tt 11}\cdots{\tt  11}}
	\theta_{\tt 1} \theta_{\tt 2} \cdots \theta_{\tt m-1}\theta_{\tt m}
	$$
	succinctly encodes elements $( x_{{\tt 000}\cdots{\tt 00}},\ldots ,  x_{{\tt 111}\cdots{\tt 11}})$ of the set $[0,1]^{2^m}$.	The question of how to write down dynamical states for dynamical systems that include discrete degrees of freedom is the main goal of the remainder of this article. 
For the simplest case of a two state discrete system, for an observable $X(\theta) = 	\frac{x_{{\tt 0}} + x_{{\tt 1}}}2 +\frac{x_{{\tt 0}} - x_{{\tt 1}}}2\hh \theta$, a state might be given as
$$
\Psi = p_{\tt 0} (1+\theta) + p_{\tt 1} (1-\theta)\, ,
$$
where the probabilities $p_{\tt 0}$ and $p_{\tt 1}$ sum to unity (the combinations of Grassmann variables appearing here are to do with the construction of pure states).
The expectation $\langle X\dangle_{\Psi}$ of~$X$ with respect to the state $\Psi$ should then be
$$
p_{\tt 0} x_{\tt 0} +  p_{\tt 1}  x_{\tt 1}
\stackrel?=\frac{\int_{\mathcal M}  X\Psi }{\int_{\mathcal M}   \Psi}\, .
$$
The natural question is whether this expectation can be expressed as a suitable integration over a symplectic supermanifold ${\mathcal M}$. We would also like to develop the corresponding notions of instantaneous laboratories and dynamics. Thus we now describe pertinent aspects of supermanifold geometry;  a more detailed account can be found in~\cite{manin1997introduction, rogers2007supermanifolds}. Many supergeometers can safely proceed directly to Section~\ref{supergeometerslandhere}.

	\subsection{Supermanifolds}\label{showkeys} An $(n,m)$ \textit{supermanifold} $\mathcal{M}$ is a pair $(M,\mathcal{A})$ where~$M$ is a smooth, dimension $n$ manifold termed the \textit{body}, and $\mathcal{A}$ is a $\mathbb{Z}_2$-graded commutative sheaf of functions on $M$; $m$ is called the \textit{odd-dimension} of~$\mathcal{M}$. By definition, there is a surjective map $\epsilon: \mathcal{A} \to C^{\infty}M$ compatible with restrictions to open sets. That is, for each open set~$U \subseteq M$, we have a surjective homomorphism $\epsilon_{U}: \mathcal{A}U \to C^{\infty}U$, where $\mathcal{A}U:=\mathcal{A}|_{U}$. Moreover~$\mathcal{A}$ satisfies a local triviality condition: given an open cover $\{U_{\alpha}\}$ of~$M$, there are local isomorphisms \begin{equation}\label{exterior}\mathcal{A}U_{\alpha} \cong C^{\infty}U_{\alpha} \otimes \Lambda \mathbb{R}^m\, .\end{equation} In physics, the sheaf of functions $\mathcal{A}$ corresponds to the space of superfields. 
	Hence we will slightly abuse notations and also use ${\mathcal A}$ to denote the space of superfunctions $\Gamma(M,{\mathcal A})$.
	These form a supercommutative algebra and the output of the map $\epsilon$ is called the body of a superfield. 
	
	Note that if $M$ is any smooth manifold, then the pair $(M, C^{\infty}M)$ defines a supermanifold with a trivial $\mathbb{Z}_2$-grading. A less trivial example is $(\mathbb{R},C^{\infty}\mathbb{R})$ where $C^{\infty}\mathbb{R}=C_+^{\infty}\mathbb{R}\oplus C_-^{\infty}\mathbb{R}$ and $\pm$ denote even and odd functions, respectively.  Another important example is the exterior bundle $\Lambda M$ for any~$M$.  	
	
	A supermanifold locally looks like the trivialization of a vector bundle. We call this the local model and denote it by $\mathbb{R}^{n|m}:=(\mathbb{R}^n,C^{\infty}\mathbb{R}^n \otimes \Lambda\mathbb{R}^m)$. Of course $\mathbb{R}^{n|m}$ is itself a supermanifold. The sheaf of functions on~$\mathbb{R}^{n|m}$ is a section space of the vector bundle $\Lambda\mathbb{R}^m\to\mathbb{R}^n \times \Lambda\mathbb{R}^m\to  \mathbb{R}^n$, {\it id est}
	$$C^{\infty}\mathbb{R}^n \otimes \Lambda\mathbb{R}^m= \Gamma(  \mathbb{R}^n \times \Lambda\mathbb{R}^m )\, .$$ 
	These can be encoded by ``functions'' $F(x,\theta)$ of 
	 $n$ bosonic/even coordinates~$\{x^i\}$ and $m$ fermi\-onic/odd/Grassmann ``coordinates'' $\{\theta^a\}$, respectively, where the latter are subject to the relation 
	\begin{equation*}
		\theta^a\theta^b+\theta^b\theta^a=0\, ,
	\end{equation*} 
	and
	$$
	F(x,\theta) := f(x) + f_a(x) \theta^a + f_{ab}(x) \theta^a \theta^b + \cdots + f(x)_{a_1\cdots a_m} \theta^{a_1}\cdots \theta^{a_m}\, .
	$$
The coefficients above are real-valued functions of ${\mathbb R}^n$.

\smallskip
	One often wants to consider \textit{supermanifold morphisms} which are defined as follows: Let $\mathcal{M}=(M,\mathcal{A})$ and  $\mathcal{N}=(N,\mathcal{B})$ be supermanifolds. A supermanifold morphism
	$$\Phi:\mathcal{M} \to \mathcal{N}$$
	 is a pair of maps $(f,\varphi)$ where $f:M \to N$ is a manifold morphism and $$\varphi : \mathcal{B}\to f_*\mathcal{A}$$ is a morphism of sheaves over $N$. The map $f_*$ is constructed by considering preimages of open sets in $N$, and is often termed the sheaf pushforward~\cite{manin1997introduction}.\\
	
	 On a supermanifold, there is a canonical, globally-defined \textit{body map} between sheaves $$i:(M,C^{\infty}M) \to \mathcal{M}\, ,$$ given by the pair $(\text{id}_M,\epsilon_{M})$. The body map   pulls back superfields and superdifferential forms (the latter will be defined below) on $\mathcal{M}$ to the body manifold $M$. The body of a superfield $F\in \mathcal{A}M$ is denoted by $F_0:=\epsilon_M(F)$ (and similarly for forms).

We are particularly interested in the case where we can describe supergeometry using vector bundles as this melds well
with our previous discussion of probabilistic states, see in  particular Definition~\ref{mylozengeisblack}.
For that we need a map projecting fibers to the base. However a  projection $$\pi:\mathcal{M} \to(M,C^{\infty}M)\: \mbox{ such that  }\: \pi\circ i =\text{id}_M$$ may not, in general, exist globally~\cite{manin1997introduction}. When it does, $\mathcal{M}$ is called a \textit{split supermanifold}. A seminal theorem due to  Batchelor guarantees that real supermanifolds are always split~\cite{batchelor1979structure}; this relies on partitions of unity for the base manifold $M$.

	A split supermanifold~$\mathcal{M}$ is isomorphic to an  exterior bundle $$\pi: \mathcal{E}M = M\ltimes \Lambda V^* \to M$$ associated to a vector bundle $$\mathcal{V}M:=M\ltimes V\to M\, .$$ Points in $\mathcal{V}M$ may be loosely thought of as  ``superpoints'' in $\mathcal{M}$, while the sections of $\mathcal{E}M$ correspond to superfields on $\mathcal{M}$. It is possible to define a precise notion of superpoints using the so-called ``functor of points''~\cite{vaquie2021sheaves}, but this will not play an important {\it r\^ole} here.
		Because of our interest in vector bundles,  in what follows always work with split supermanifolds.  \\
	
	The isomorphism between a  split supermanifold $\mathcal{M}$ and a vector bundle~$(\mathcal{E}M,\pi)$  cannot be realized canonically; a given choice of such is termed a \textit{splitting}. We will label  splittings by their corresponding projection $\pi$. Certain supermanifold (auto)morph\-isms~$\mathcal{M}\to \mathcal{M}$  neither preserve the projection $\pi$ nor correspond to bundle automorphisms of $\mathcal{E}M$. Moreover,  the vector bundle $\mathcal{E}M$ has an
	 additional $\mathbb{Z}_{m+1}$-grading 
	because it is an associated exterior bundle of $\mathcal{V}M$. This grading is not preserved by such morphisms. 
While a supermanifold $\mathcal{M}$ itself need not have a canonical~$\mathbb{Z}_{m+1}$-grading, there is a canonical (splitting independent) descending filtration. This is described below, but first we develop the vector bundle picture and discuss superdiffeomorphisms.

\smallskip
	
	The body $M$, defined as the quotient of the total space~$\mathcal{E}M$ by the projection~$\pi$, is independent of any splitting. Because we want to work independently of the choice of a splitting, we could instead consider the following quotient 
	\begin{equation*}
		\mathcal{M}:=\Bigg( \bigsqcup_{ \{ \Phi: \mathcal{M}\to \mathcal{M} \} }(\mathcal{E}M,\pi\circ\Phi) \Bigg) \Bigg/ \sim\, ,
\end{equation*}
and take this as a  working definition of a supermanifold. In the above, the quotient is by the equivalence relation $(\mathcal{E}M,\pi_1) \sim (\mathcal{E}M,\pi_2)$ if and only if~$\pi_1=\pi_2\circ\Phi$ for any supermanifold automorphism $\Phi$ on $\mathcal{M}$. In particular we focus on the case that $\mathcal{M}$ is a smooth supermanifold, and quite generally  assume that all structures are smooth. In this case, we term such supermanifold automorphisms, \textit{superdiffeomorphisms} of $\mathcal{M}$.
Note that by construction $\mathcal{M}$ is itself a vector bundle with base~$M$ and projection obtained by suitably quotienting a disjoint union of projections. \\

For computations, we may always choose a particular representative~$(\mathcal{V}M,\pi)$, or equivalently $(\mathcal{E}M,\pi)$; observe this does not rely on a particular choice of local coordinates. On the other hand, on a trivializing chart with coordinates $x^i$ on $M$, let us denote the fiber coordinates of~$\mathcal{V}M$ and~$\mathcal{E}M$ by $y^a$ and $\theta^a$, respectively. Then, on intersecting trivializing covers, the respective  bundle automorphisms take the form 
\begin{equation}
	\tilde{x}^i=\tilde{x}^i(x)\, , \quad \tilde{y}^a=g^a_{\ b}(x)y^b\, , \quad \tilde{\theta}^a=g^a_{\ b}(x)\theta^b\, , \quad g \in GL(m,\mathbb{R})\, . \label{glue}	
\end{equation}
Here, the fiber coordinates of~$\mathcal{V}M$ are Grassmann even, whereas that of $\mathcal{E}M$ are odd. This is an example of the parity reversal operation $\Pi$, that acts as $\Pi y^a = \theta^a$, and changes the Grassmann even grading of the coordinates $y^a$ to the  odd one of $\theta^a$. In some contexts, it is useful to think of these odd coordinates as covectors~$\theta^a=dy^a$. \\

In general, the \textit{parity reversal} $\Pi$ is defined  on a $\mathbb{Z}_2$-graded vector space $V=V_0 \oplus V_1$ by
\begin{equation*}
	(\Pi V)_0 = V_1\, ,\qquad  (\Pi V)_1 = V_0 \, , 
\end{equation*}
where subscripts $0,1$ denote  even and odd subspaces, respectively. In particular, if~$V$ is purely even, \textit{id est} $V=V_0$, then  polynomial functions on $ \Pi V$ can be identified with elements of the exterior algebra of the dual $V^*$, that is 
\begin{equation*}
	\mathbb{R}[\Pi V] \cong \Lambda V^*\, . 
\end{equation*}
Consequently, by applying the parity reversal fiberwise to the vector bundle $\mathcal{V}M$ and considering functions on the total space, we end up with the section space of $\mathcal{E}M$. This observation will be important when we discuss super vectors on $\mathcal{M}$ and their relation to derivations on $\Pi \mathcal{V}M$.

Given data  such as a particular superfield on a supermanifold, we may pick a representative~$(\mathcal{E}M,\pi)$, and then consider bundle
 invariant expressions (respecting~\nn{glue}) built therefrom. 
 In turn we can
  explicitly compute the superdiffeomorphism transformation of such expressions in local coordinates $(x,\theta)$ to verify whether they are supermanifold invariants. For instance any superfield 
 $F\in \Gamma(\mathcal{E}M)$ has a decomposition
\begin{equation*}
	F(x,\theta)= F_0(x)+\hat{F}(x,\theta)\, ,
\end{equation*}
where
\begin{equation*}
	\hat{F}=\sum_{a}\theta^{a}\hat{F}_{a}(\theta)=:O(\theta)\, .
\end{equation*}
Here the  notation 
$O(\theta^k)$ denotes terms of polynomial degree at least $k\in\{0,\ldots,m\}$ in the odd coordinates $\theta$. In these coordinates the invariantly defined body of $F$ is given by   $F_0(x)\in C^{\infty}M$ while the \textit{soul} $\hat{F}$  depends on the choice of splitting $\pi$. Superdiffeomorphisms themselves  are locally represented by a set of superfields
\begin{equation}
	\tilde{x}^i=\tilde{x}^i(x,\theta)\, , \quad  \tilde{\theta}^a= \tilde{\theta}^a(x,\theta)\, . \label{superdiff}
\end{equation}
Note that the new coordinates $(\tilde{x},\tilde{\theta})$ may have  more general dependence on the original  ones than that of the bundle gluings~\nn{glue}. 
Also, by dint of invertibility, $\tilde{x}$ must be an even polynomial of $O(\theta^0)=:O(1)$, and $\tilde{\theta}$ is an odd polynomial of $O(\theta)$. \\

It follows from the structure of superdiffeomorphisms described in~\nn{superdiff}, that the algebra of superfields~$\mathcal{A}$, and hence the supermanifold $\mathcal{M}$ itself, has a 
decreasing filtration:
\begin{equation}\label{filter}
	\mathcal{A}=\mathcal{A}_0 \supset \mathcal{A}_1 \supset \cdots \supset \mathcal{A}_m \, ,
\end{equation}
where $m$ is the odd dimension of $\mathcal{M}$, and in any choice of coordinates $\mathcal{A}_k = O(\theta^k)$. Observe that given a superfield $F \in \mathcal{A}_k$, then $F_k \in \mathcal{A}_k/ \mathcal{A}_{k+1}$ is invariantly defined. 
The body~$F_0$ of $F \in \mathcal{A}$ is a particular example.
 A choice of splitting $(\mathcal{E}M,\pi)$ further induces a $\mathbb{Z}_{m+1}$-grading associated with this filtration. Indeed, the choice of a splitting is equivalent to such a~$\mathbb{Z}_{m+1}$-grading~\cite{koszul1994connections}.

\medskip

\textit{Super vectors} will be central to 
our study of dynamics; these are defined as sections of the \textit{tangent sheaf} $\hh T\mathcal{M}$. The latter is  the subsheaf of graded derivations of the sheaf of endomorphisms of $\mathcal{A}$. Super vector fields form a rank $n+m$, free $\mathcal{A}M$-module. 
They also 
form a Lie superalgebra under the \textit{supercommutator} (inherited from the usual composition of endomorphisms) of super vector fields denoted by $[\pdot,\pdot\}$. We will also need superdifferential forms. \textit{Super covectors} are sections of the \textit{cotangent sheaf}~$T^*\! \mathcal{M}$. The latter is defined as the pointwise $\mathcal{A}M$-dual of the tangent sheaf $T\mathcal{M}$. Sections of $T^* \! \mathcal{M}$ are called \textit{superdifferential one-forms}; these also form a rank $n+m$, free~$\mathcal{A}M$-module. Higher forms are defined below.

To describe super vectors and covectors/one-forms on a representative vector bundle~$(\mathcal{E}M,\pi)$, we  denote the tangent sheaf
by $T_\pi \! \hh \mathcal{M}$. Note in particular that this does not equal the space $T\mathcal{E}M$. Similarly we denote the cotangent sheaf by $T^*_\pi \! \hh \mathcal{M}$. On a trivializing patch with coordinates $(x,\theta)$, super vectors and covectors are spanned by $\{\frac{\partial}{\partial x}, \frac{\partial}{\partial \theta}\}$, and their duals $\{dx, d\theta\}$, respectively, with coefficients taking values in sections of $\mathcal{E}M$. A super vector $v \in T_\pi \! \hh \mathcal{M}$ and a super covector $\alpha \in T^*_\pi \! \hh \mathcal{M}$ are then locally expressed as
\begin{align*}
v &= v^i(x,\theta)  \frac{\partial}{\partial x^i} + v^a(x,\theta) \frac{\partial}{\partial \theta^a}\, , \\
\alpha &= \alpha_i(x,\theta) dx^i + \alpha_a(x,\theta) d\theta^a\, .  	
\end{align*} 
Vectors $\frac{\partial}{\partial x}$ along the base are Grassmann even, and vectors $\frac{\partial}{\partial \theta}$ along the fiber are Grassmann odd. Note that on the body $f(x)\in C^{\infty} M$ of some superfield, vectors $\frac{\partial}{\partial x}$ act as derivations in the standard way. Also note that the vector $v^a\frac{\partial}{\partial \theta^a}$ acts as interior multiplication on the fiber viewed as the exterior algebra $\Lambda V^*$. The coordinate one-form $dx$  is  assigned odd Grassmann parity while $d\theta$ is even  so that, for $\alpha$ and $v$ Grassmann even,
\begin{equation*}
\alpha(v)= \alpha_i v^i +\alpha_a v^a\in \Gamma({\mathcal E}M)\, .
\end{equation*}

For even super vectors $v$, there is a well-defined  notion of a body $v_0$. Namely, for any function $f_0$ on the body, 
$$
v_0 f_0:= (vf)_0\, ,
$$  
where $f$ is any superfunction with body $f_0$. 

There also exists a general notion of superflows along super vectors, see for example~\cite{monterde1993existence,garnier2013integration}. We only require a simple version of this:
Given an even supervector $v\in \Gamma(T\mathcal{M})$ and a chart $(U,\mathcal{A}U)$, for short times $t\in (-\delta,\delta)\subset{\mathbb R}$, we can define a supermanifold morphism between charts $$\exp(t v):(U,\mathcal{A}U)\to (U_t,\mathcal{A}{U_t})\,,$$
by requiring the following properties: $(i)$ $\exp(0 v)=\operatorname{Id}$, (ii) $\exp\big((t+t')v\big)=\exp(tv)\circ \exp(t'v)$, and $(iii)$ given any  superfunction~$f$ and any
element $X$ of the dual ${\mathcal A}^*$ to  the space of superfunctions, 
$$
\frac{dX(f_t)}{dt}= X(v f_t)\, ,
 $$
 where $f_t:= \exp(-tv)^* f$ is defined by pulling back the  second map ${\mathcal A}U_t\to {\mathcal AU}$ of the supermanifold morphism. 
 Notice that the map $U\to U_t$ is defined by the standard exponential map
 $\exp(t v_0)$. Short time existence and uniqueness of $\exp(tv)$ is established in~\cite{garnier2013integration}. Formally, $\exp(tv)$ acts on functions according to Taylor expansion in powers of derivations.
%
%
%
%
%
%
%
%
%
%
%
%
%
%
%
%
%
%
In local coordinates $v = v^i(x,\theta)  \frac{\partial}{\partial x^i} + v^a(x,\theta) \frac{\partial}{\partial \theta^a}$, the open set 
$U_t:=\exp (tv_0^i(x)\frac{\partial}{\partial x^i})(U)$ and $\mathcal{A}U \ni f(x^i,\theta^a) \mapsto f(x^i+tv^i,\theta^a+t v^a)$. For suitable $v$, this (short time) map extends to all of $\mathcal{M}$; main properties of the usual exponential map hold for the super case too because it defines a one-parameter family of superdiffeomorphisms~\cite{jetzer1999completely}.\smallskip

On a representative vector bundle $(\mathcal{E}M,\pi)$, there are invariant notions of vertical super vectors and horizontal super covectors. Sections of the vertical subbundle~$\text{Ver} \hh T_\pi \! \hh \mathcal{M} := \text{ker}\hh d\pi\subset T_\pi \! \hh \mathcal{M}$ of the tangent bundle are termed \textit{vertical super vectors}. The fibers of $\text{Ver} \hh T_\pi \! \hh \mathcal{M}$ are isomorphic to that of~$\Pi\mathcal{V}^*M$. Furthermore, a super covector $\alpha$ satisfying~$\alpha(v)=0$ for any vertical super vector $v$, is termed a \textit{horizontal super one-form}. In fact we have an exact sequence
\begin{equation*}
	0 \rightarrow \Gamma(\mathcal{E}M\otimes \text{Ver} \hh T_\pi \! \hh \mathcal{M}) \rightarrow \Gamma(T_\pi \! \hh \mathcal{M}) \xrightarrow{d\pi} \Gamma(\mathcal{E}M\otimes TM) \rightarrow 0\, .
\end{equation*}
Super covectors decompose according to the dual exact sequence. \\

A {\it superaffine connection} $\nabla$ on a supermanifold is a map
$$
\nabla:\Gamma(T{\mathcal M})
\to 
\Gamma(T^*{\mathcal M}\otimes T{\mathcal M})
$$
subject to superanalogs of the usual properties:
$$
\nabla_{f U} V = f \nabla_U V\, ,\qquad
\nabla_U (f V) = (U f) V + f \nabla_U V\, ,
$$
for any (even) superfield $f$ and supervectors $U$ and $V$. 
A seminal result of Koszul shows that a superaffine connection determines a splitting~\cite{koszul1994connections}.
For this let $F\in {\mathcal A_k}$ (see Equation~\nn{filter}).
Then we search for an even  super vector~$X$ such that
$$
X F = k F + G
$$
for some $G \in {\mathcal A_{k+1}}$. Note that for ${\mathbb R}^{n|m}$ the ``Eulerian'' vector 
$X= \theta^a \frac\partial{\partial \theta_a}$ obeys the above stipulation. Koszul shows that upon additionally demanding 
$$
\nabla_X X = X\, ,
$$
the vector $X$ both exists and is unique. Then one can study the space of odd super vectors $U$ such that 
$$
\nabla_X U = -U\, .
$$
The exponential map of $\nabla$ in  these directions determines vertical fibers of the splitting associated to $\nabla$. Eigenvalues of the vector $X$  yield a ${\mathbb Z}_{m+1}$-grading on sections of  the corresponding vector bundle. We term the latter the {\it Koszul bundle of $\nabla$}, or a {\it Koszul splitting}.

\medskip

Arbitrary rank tensor fields on $\mathcal{M}$ are defined by considering tensor products of $T\mathcal{M}$ and $T^* \! \mathcal{M}$ over $\mathcal{A}M$. In particular \textit{superdifferential forms} are defined by graded skew tensor products of $T^* \! \mathcal{M}$ whose section space is denoted by $\Omega \mathcal{M}$. Note that while superdifferential forms can only depend polynomially on the Grassmann odd differentials~$dx$, the same is not true for the even differentials $d\theta$. In principle, the precise functional dependence on the $d\theta$ depends on the physical problem of interest, see for example~\cite{fuchscohomology, bastianelli2009detours}. Here we focus on formal power series. \\

The space of superdifferential forms is graded by form degree: 
\begin{equation}\label{Omega}\Omega \mathcal{M}= \bigoplus_{k \in \mathbb{N}} \Omega^k \mathcal{M}\, . 
 \end{equation}
In local coordinates, the section spaces $\Omega^k \mathcal{M}$ are eigenspaces of the Euler operator $dx^i \frac{\partial}{\partial dx^i}+d\theta^a\frac{\partial}{\partial d\theta^a}$ (some authors denote $\frac{\partial}{\partial dx^i}$ by $\iota_{\frac{\partial}{\partial x^i}}$ and $\frac{\partial}{\partial d\theta^i}$ by $\iota_{\frac{\partial}{\partial \theta^i}}$). There is also a $\mathbb{Z}_2$-grading by the total Grassmann parity.

\subsection{Differential Calculus.} For a representative vector bundle $(\mathcal{E}M,\pi)$, we denote the sheaf of superdifferential forms by $\Omega_{\pi} \! \hh \mathcal{M}$. The grading determined by the choice $(\mathcal{E}M,\pi)$ allows us to consider separately the operators $dx^i \frac{\partial}{\partial dx^i}$ and $d\theta^a\frac{\partial}{\partial d\theta^a}$, with respective eigenvalues $p$ and $q$. This gives a further decomposition on top of that of Equation~\nn{Omega}, 
$$\Omega^k_{\pi} \! \hh \mathcal{M} =\bigoplus_{p+q=k, \ p,q\geq0} \Omega^{p,q}_{\pi} \!  \mathcal{M}\, .$$ In local coordinates $\Omega^{p,q}_{\pi} \!  \mathcal{M}$ is the space of $(p+q)$-forms with $p$ bosonic and $q$ fermionic generators $dx$ and $d\theta$, respectively. \\

The exterior derivative, interior product, and Lie derivative of superdifferential forms can be defined along similar lines to their bosonic counterparts. 
 One may define these in local coordinates for a representative vector bundle and then verify that this defines superdiffeomorphism invariant operations. In particular the exterior derivative $d_T:\Omega^k_{\pi} \! \hh \mathcal{M} \to \Omega^{k+1}_{\pi} \! \hh \mathcal{M}$ is defined by
\begin{equation}\label{dpdh}d_T := dx^i \frac{\partial}{\partial x^i}+ d\theta^a\frac{\partial}{\partial \theta^a}\, .  \end{equation}
It is a nilpotent operator, {\it videlicet}  $d_T^2=0$. In calculations we often denote the respective summands on the right hand side above by $d$ and~$\hat{d}$. The choice of splitting thus leads to   maps $\Omega^{p,q}_{\pi} \!  \mathcal{M} \stackrel d\to \Omega^{p+1,q}_{\pi} \!  \mathcal{M}$ and $\Omega^{p,q}_{\pi} \!  \mathcal{M} \stackrel{\hat d}\to \Omega^{p,q+1}_{\pi} \!  \mathcal{M}$ for the corresponding representative $\mathcal{E}M$. \\

Given an even super vector field $v\in \Gamma(T_\pi \! \hh \mathcal{M})$, the {\it interior product} maps $\Omega^k_{\pi} \! \hh \mathcal{M}$ to~$\Omega^{k-1}_{\pi} \! \hh \mathcal{M}$ according to 
\begin{equation*}
	\iota_v:= v^i \frac{\partial}{\partial dx^i} +v^a\frac{\partial}{\partial d\theta^a}\, .
\end{equation*} 
Once again this defines a pair of maps $\iota_{v^i \frac{\partial}{\partial x^i}}$ and $\iota_{v^a \frac{\partial}{\partial \theta^a}}$ with domains $\Omega^{p,q}_{\pi} \!  \mathcal{M}$. The Lie derivative with respect to a super vector $v$ acting on superdifferential forms is defined, {\it \`a la} Cartan's magic formula, by
$$\mathcal{L}_v:=d_T \hh \iota_v+\iota_v \hh d_T\, .$$  

\smallskip
Finally, for any supermanifold  ${\mathcal M}$ , there is a canonical, top  degree polyvector 
$$
{\mathscr X}\in \Gamma(\Lambda^m T {\mathcal M})\, , 
$$
that we term  the {\it Eulerian top field}: Given a splitting and  corresponding representative vector bundle $\mathcal{E}M$, there is a canonical 
 fiber-density valued, top polyvector
\begin{equation} \label{toppolyvec}
	\epsilon:=\frac{1}{m!}\hh\epsilon^{a_1\cdots a_m} \frac{\partial}{\partial \theta^{a_1}} \cdots \frac{\partial}{\partial \theta^{a_m}} \in \Gamma\big(  (\Lambda^m\mathcal{V}_\pi M)^{-1}\otimes \Lambda^m \Pi \mathcal{V}^*_\pi M\big)\, ,
\end{equation}	  
where $\epsilon^{a_1\cdots a_m}$ denotes the totally anti-symmetric Levi-Civita symbol. The above is not canonical to the supermanifold ${\mathcal M}$ itself. Nor is the fiberwise Euler operator
$$\theta^a\frac{\partial}{\partial \theta^a} \in \Gamma(\Pi \mathcal{V}_\pi M\otimes \Pi \mathcal{V}_\pi^*M)\, ,$$
(whose eigenvalues correspond to  the ${\mathbb Z}_{m+1}$ grading {\it \`a la} Koszul). However, combining these  ingredients---taking  
 the top skew power of the fiber Euler operator---does give an invariant object
\begin{equation}  \label{eulertop}
	\X:=\frac{1}{m!} \theta^{a_1} \cdots \hh \theta^{a_m}\frac{\partial}{\partial \theta^{a_1}} \cdots \frac{\partial}{\partial \theta^{a_m}} \in \Gamma( \Lambda^m\Pi\mathcal{V}\otimes \Lambda^m \Pi \mathcal{V}^* )
\subset \Gamma(\Lambda^m T {\mathcal M})	\, .
\end{equation}
%

\subsection{Integral Calculus.} Let us begin with the supermanifold $\mathbb{R}^{n|m}$. A superfield $F \in \mathcal{A} \mathbb{R}^{n|m}$ is called \textit{super integrable} if its coefficients in a $\theta$-expansion are   $\mathbb{R}^n$-integrable. Its \textit{Berezin integral} is given by~\cite{Berezin} 
\begin{equation*}
	\int_{\mathbb{R}^{n|m}} F  := \int_{\mathbb{R}^n}  d^nx \frac{\partial}{\partial \theta^m} \cdots \frac{\partial}{\partial \theta^1} F(x,\theta) \, .
\end{equation*}
In the case $n=m$, the supermanifold $\mathbb{R}^{n|n}$ can be viewed as the vector bundle $\Lambda \mathbb{R}^n\to {\mathbb R}^n$. A superfield $F$ then amounts to a (possibly inhomogeneous) differential form $F \in \Omega\mathbb{R}^n$. The Berezin integral computes the integral over $\mathbb{R}^n$ of its top form component. \\

The Berezin integral depends on an $\mathbb{R}^{n|m}$-integration measure $\mu=d^nx \frac{\partial}{\partial \theta^m} \cdots \frac{\partial}{\partial \theta^1}$. Under a superdiffeomorphism $\varphi: (x,\theta)\mapsto (\tilde{x},\tilde{\theta})$, this transforms by a (right) factor of the \textit{superdeterminant} defined by 
\begin{equation*}
	\operatorname{sdet}(\varphi) =  \frac{\operatorname{det}(\varphi_{00}-\varphi_{01}\varphi_{11}^{-1}\varphi_{10}^T)}{\operatorname{det}(\varphi_{11})}=\frac{\operatorname{det}(\varphi_{00})}{\operatorname{det}(\varphi_{11}-\varphi_{10}^T\varphi_{00}^{-1}\varphi_{01})}\, , 
\end{equation*}
where the above matrices are $\varphi_{00}:= \frac{\partial x}{\partial \tilde{x}}$, $\varphi_{01}:= \frac{\partial x}{\partial \tilde{\theta}}$, $\varphi_{10}:= \frac{\partial \theta}{\partial \tilde{x}}$, and $\varphi_{11}:= \frac{\partial \theta}{\partial \tilde{\theta}}$. \\

On a general supermanifold $\mathcal{M}$, one can integrate sections of a distinguished line bundle $\text{Ber}\mathcal{M}$ called the \textit{Berezinian} of $\mathcal{M}$~\cite{rogers2007supermanifolds}. On a representative $(\mathcal{E}M,\pi)$, the  Berezinian admits the decomposition
\begin{equation*}
	\text{Ber}_{\pi} \mathcal{M}= \Lambda^n_\pi M \otimes \Lambda^m \Pi\mathcal{V}^*_\pi M\, .
\end{equation*}

Locally, a section of $\text{Ber}_{\pi} \mathcal{M}$ is generated by $d^nx \frac{\partial}{\partial \theta^m} \cdots \frac{\partial}{\partial \theta^1}$
in concordance with the Berezin measure $\mu$ on $\mathbb{R}^{n|m}$. Thus, under a superdiffeomorphism,  a section of  the Berezinian transforms by a factor of the superdeterminant of the Jacobian.
Berezin integrals can then be performed
by trivializing sections of the Berezinian
over
a partition of unity; see for example~\cite{manin1997introduction, rogers2007supermanifolds}.

%
%


\subsection{Hermitean Superfields}
\label{Hermitean}

\medskip
Given a choice of splitting, there is a notion  of a hermitean superfield (see for example~\cite{schmitt1990supergeometry}) defined as follows. 
First one considers the space of complex-valued superfields
$$
F\in \Gamma({\mathcal E}M) \oplus i \Gamma({\mathcal E}M)=:\Gamma_{\mathbb C}({\mathcal E}M) \, .
$$
Then we  define an (antilinear) involution $\dagger :\Gamma_{\mathbb C}({\mathcal E}M) \to \Gamma_{\mathbb C}({\mathcal E}M)$ via the map
$$
F^\dagger=\big(\oplus_{k=0}^m F_k)^\dagger
:=\oplus_{k=0}^m (-1)^{\lfloor\frac k2 \rfloor} F^*_k\, ,
$$
where the direct sum above is that of the ${\mathbb Z}_{k+1}$ grading and 
$*$ denotes complex conjugation. 
In local (real) Grassmann coordinates $\theta^a$, the above amounts to defining
$$
(\theta^a)^\dagger = \theta^a\: \mbox{ and } \: (\theta^a \theta^b)^\dagger := (\theta^b)^\dagger (\theta^a)^\dagger = - \theta^a \theta^b\, . 
$$
A complex-valued vector field that obeys
$$
F=F^\dagger
$$
is termed a {\it hermitean superfield}. The above involution induces corresponding actions on super tensors where the grading is obtained by counting  $\theta$'s, $d\theta$'s and $\frac\partial{\partial\theta^a}$'s. For example, $(d\theta^a\wedge  d\theta^b)^\dagger = -d\theta^a\wedge  d\theta^b$. This ensures that a hermitean super $k$-form integrated along a $k$-dimensional manifold obeys a reality condition.
%
%
\begin{example}
Consider the 
spinning particle model  
 of Berezin and Marinov~\cite{Berezin:1976eg}: 
$$
	S[\theta(t)]=  i \int dt (\delta_{ab}\theta^a \dot{\theta^b} +\epsilon_{abc}v^a \theta^b \theta^c) \, ,
$$	
where $\theta^a\in {\mathbb R}^{0|3}$
and $v^a\in  {\mathbb R}^{3}$. Since $(\theta^a \dot\theta^b)^\dagger = -\theta^a \dot\theta^b$, the above action is ``real''. Hence the one-form
$
i\big(\delta_{ab}\theta^a d{\theta^b} + \epsilon_{abc}v^a \theta^b\theta^c dt)
$
is hermitean.
 \hfill$\blacksquare$
\end{example}

\smallskip

Note that in physics the above notion of hermiticity is called a real superfield.
Observe however that  the product of two hermitean superfields is not necessarily hermitean (for example $(\theta^1 \theta^2)^\dagger = -\theta^1 \theta^2$), so the subset of hermitean superfields in $\Gamma_{\mathbb C}({\mathcal E}M)$ is not closed under multiplication. One can also consider supermanifolds given by a sheaf of complex exterior algebras $\Lambda^m {\mathbb C}$ over a base manifold $M$ and extra data of an involution~$\dagger$~\cite{schmitt1990supergeometry}. We prefer to work with the above, real supermanifold-based notion of hermiticity. 
Hermitean superfields will play an important {\it r\^ole} when defining an inner product in Section~\ref{inner}. 

\subsection{Symplectic Supermanifolds}\label{supergeometerslandhere} A \textit{symplectic supermanifold} $(\mathcal{M},\Omega)$ is an $(n,m)$-supermanifold~$\mathcal{M}$ equipped with a maximally non-degenerate, closed, Grassmann-even, superdifferential two-form $\Omega$. 
By maximally non-degenerate we mean that the subbundle of $T\mathcal{M}$ defined by the super-vector solution space of
\begin{equation*}
\iota_v \Omega=0\, , 
\end{equation*}
 has rank zero or one, as a module over the ring of superfunctions on~${\mathcal M}$. Equivalently,  the subbundle $\operatorname{ker}\flat \subset T\mathcal{M}$ defined by the kernel of the musical map
\begin{align*}
	\flat: T\mathcal{M} &\to T^*\mathcal{M} \, ,\\[-1mm]
\rotatebox{90}{$\scriptscriptstyle \in$}\:	&\quad\quad
\rotatebox{90}{$\scriptscriptstyle \in$}\\[-1mm]	
		v \;&\mapsto\; \iota_v \Omega 
\end{align*}
is Grassmann even and has rank either zero or one. 
For physical applications we will need supermanifolds equipped with a hermitean symplectic form. Again these are classified by the rank of the musical map $\flat$ (but now over ${\mathbb C}$). To  
 distinguish the two cases, following Section \ref{sec3}, we term supermanifolds equipped with a hermitean symplectic form with even or odd dimensional body manifolds, \textit{super phase-spaces} and \textit{super phase-spacetimes}, {and typically denote them by $({\mathcal M},\Omega)$ and
 $({\mathcal Z},\Omega)$, 
 respectively. 
Exactly as in the bosonic case, a hermitean supersymplectic form $\Omega$ on a super phase-space is non-degenerate, whereas on a super phase-spacetime it admits a one dimensional nontrivial (Grassmann even) kernel. \\ 

\begin{example}[Rothstein structures]\label{Rothstein}
Rothstein~\cite{rothstein1991structure} shows that real (non-degenerate) super symplectic manifolds correspond in a natural way to the data  $({\mathcal V}M,\eta,\nabla,\omega)$ of a vector bundle~${\mathcal V}M$, 
a bundle metric $\eta$ with a compatible bundle connection $\nabla$ as well as a symplectic form $\omega$ on the base.
 A careful account of Rothstein structures necessitates a discussion of the Ehresmann description of bundle connections on $T{\mathcal M}$. These details are not needed in a local coordinate description. We focus on a special case due to Sasaki~\cite{sasaki1958differential}: Let~$(M,g)$ be an even dimensional  Riemannian manifold and ${\mathcal M}=(M,\Omega M)$.
 In local coordinates~$x^{\mu}$ for  $M$, the fibers of ${\mathcal M}$ are generated by (odd) covectors $\theta^{\mu}=dx^{\mu}$. Superfields on~${\mathcal M}$ are differential forms  on $M$.  
 
 Let us define an operator mapping $\Omega {\mathcal M}\to \Omega{\mathcal M}$ by
 \begin{equation*}
		{\nabla}:=d- \Gamma^{\mu}_{\ \nu} \theta^{\nu} \frac{\partial}{\partial \theta^{\mu}}\, ,
	\end{equation*}
	where  $\Gamma^\mu{}_\nu$ denotes  the (one-form valued) Levi-Civita connection coefficients of $g_{\mu\nu}$ and~$d:=dx^\mu\frac{\partial}{\partial x^\mu}$.
 Accordingly we can covariantize the exterior derivative operator
	\begin{equation*}
		d_T:=d+ d\theta^\mu \frac\partial{\partial \theta^\mu}={\nabla}+(\nabla \theta^{\mu}) \frac{\partial} {\partial \theta^{\mu}}\, ,
	\end{equation*}
where $\nabla\theta^\mu :=  d\theta^{\mu}+\Gamma^{\mu}_{\ \nu} \theta^{\nu}$.
Now let $\omega_0$ be any symplectic form on $M$. Then the above data can be packaged as a supersymplectic two-form
$$
\Omega=\omega_0 + g_{\mu\nu}\nabla\theta^{\mu}\nabla\theta^{\nu}-R_{\mu\nu}\theta^{\mu}\theta^{\nu}\, ,
$$
where indices are manipulated using $g_{\mu\nu}$ and the (two-form valued) Riemann-tensor $R^\rho{}_\sigma:= d \Gamma^\rho{}_\sigma+ \Gamma^\rho{}_{\kappa}\Gamma^{\kappa}_{\sigma}$.
Non-degeneracy and closed-ness of $\Omega$ are easily verified.
In particular, locally we have that $\omega = d\lambda$ and in turn
$$
\Omega = d_T \big(\lambda + \theta_\mu \nabla \theta^\mu\big)\, .
$$
Note that Rothstein has a stronger result constructing local Darboux models for super phase-spaces~\cite{rothstein1991structure}. Also note that $\omega_0 + i(g_{\mu\nu}\nabla\theta^{\mu}\nabla\theta^{\nu}-R_{\mu\nu}\theta^{\mu}\theta^{\nu})$ 
defines a hermitean symplectic form and, in turn, a super phase-space.\hfill$\blacksquare$
\end{example}

It is useful to further analyze the closure condition on the supersymplectic form~$\Omega$.
On a representative vector bundle $(\mathcal{E}M,\pi)$, we have the decomposition
\begin{equation*}
\Omega=\omega+A+\eta \, ,
\end{equation*}
where $\omega \in \Omega^{2,0}_{\pi} \!  \mathcal{M}$, $A \in \Omega^{1,1}_{\pi} \!  \mathcal{M}$ and $\eta \in \Omega^{0,2}_{\pi} \!  \mathcal{M}$, respectively. The body $\omega_0$ of~$\Omega$ is a two-form on $M$. We will primarily be interested in the generic case where $\omega_0$ itself is a symplectic form on $M$. Therefore in the case that $\Omega$ is degenerate, any degeneracy is  encoded in the symplectic body manifold $(M,\omega_0)$ as $\omega_0$ itself must be maximally non-degenerate. Because~$\Omega$ is an even supermatrix, it is invertible if and only if both~$\omega$ and~$\eta$ are invertible. Equivalently the symmetric bilinear form $\eta$ on the Grassmann distribution is always invertible.
Decomposing $d_T= d+\hat{d}$  in concord with Equation~\nn{dpdh},  closedness~$0=d_T\Omega$ amounts to  a quartet of equations:
\begin{align}
&0=d\omega \in \Omega^{3,0}_{\pi} \!  \mathcal{M}\, , 
\qquad 0=\hat{d}\omega+dA \in \Omega^{2,1}_{\pi} \!  \mathcal{M}\, , \quad \nonumber \\[-2.7mm]
 \label{quartet1}\\[-2.7mm]
&0=\hat{d} A +d\eta \in \Omega^{1,2}_{\pi} \!  \mathcal{M}\, , 
\qquad 0=\hat{d}\eta \in \Omega^{0,3}_{\pi} \!  \mathcal{M}\, .\nonumber
\end{align}
In turn, closedness of $\Omega$ implies the following Hodge-like decomposition (see~\cite{Cast}).
\begin{theorem} \label{decomposition}
Let $(\mathcal{M},\Omega)$ be a symplectic supermanifold and $(\mathcal{E}M,\pi)$ be a representative vector bundle such that $\Omega=\omega+A+\eta$ where $\omega \in \Omega^{2,0}_{\pi} \!  \mathcal{M}$, $A \in \Omega^{1,1}_{\pi} \!  \mathcal{M}$ and $\eta \in \Omega^{0,2}_{\pi} \!  \mathcal{M}$. Then \begin{equation}
	\Omega=\omega_0+\frac{1}{1+\hat{d}^{-1}d}\, \hat{d}(\beta+\gamma) \, ,  \label{sympdec}
\end{equation}
where the body 	$\omega_0\in \Omega^2M$ of $\omega$ is $d$-closed, $\hat{d}^{-1}$ is any partial inverse of $\hat{d}$, and~$\beta \in \Omega^{1,0}_{\pi} \!  \mathcal{M}$ and $\gamma \in \Omega^{0,1}_{\pi} \!  \mathcal{M}$ are arbitrary Grassmann-odd one-forms. In the above, the operator $1/(1+\hat{d}^{-1}d)$ acting on a section of $\Omega^{p,q}_{\pi} \!  \mathcal{M}$ stands for  the formal power series $\sum_{\ell=0}(-\hat{d}^{-1}d)^{\ell}$ where $
\hat d^{\hh-1}$  is defined to  be the zero operator when acting on $\Omega^{k,0}_\pi{\mathcal M}$.
\end{theorem}

The proof of the above theorem relies on the following standard (see for example~\cite{HT})  lemma and its corollary.
\begin{lemma} \label{biglemma}
Any $\hat{d}$-closed form is $\hat{d}$-exact. \label{grclosed-exact}
\end{lemma}
\begin{proof} 
It suffices to  establish the claim for a $\hat{d}$-closed form $f \in \Omega^{r,p}_{\pi} \!  \mathcal{M}$. Thanks to the $\mathbb{Z}_{m+1}$-grading on $\mathcal{E}M$, we may   focus on the coefficient of $f_{p,q}$ with $q$ Grassmann generators. 
Because $\hat d$ is a fiber-wise operator, we can employ
local coordinates for which
\begin{equation*}
f_{p,q}=
	f(x,dx)_{a_1\cdots a_qb_1\cdots b_p}\theta^{a_1}\cdots\theta^{a_q}d\theta^{b_1}\cdots d\theta^{b_p} \, .
\end{equation*}
Here the coefficients are totally anti-symmetric in the $a$ indices and totally symmetric in the $b$ indices. In Young tableaux notation, we can decompose $f$ as 
\begin{equation*}
\ytableausetup{aligntableaux=center,smalltableaux,boxsize=1em} \ytableaushort{{\sss a_1}, {\sss a_2}, {\none[\sss\vdots]}, {\sss a_q} } \: \otimes \: \ytableaushort{{\sss b_1} {\sss b_2} {\none[\sss\cdots]} {\sss b_p}} = \ytableaushort{{\sss b_1} {\sss b_2} {\none[\sss\cdots]} {\sss b_p},{\sss a_1},{\none[\sss \vdots]}, {\sss a_q} } \:\oplus\: \ytableaushort{{\sss a_1} {\sss b_1} {\none[\sss\cdots]} {\sss b_p},{\sss a_2},{\none[\sss\vdots]}, {\sss a_q} }\ \ .
\end{equation*}
Closedness of $f$ says that right tableau in the last equality vanishes. 

Now consider $f_{p-1,q+1} \in  \Omega^{r,p-1}_{\pi} \!  \mathcal{M}$ with $q+1$ Grassmann generators.
Then the exact form $\hat{d}f_{p-1,q+1} \in \Omega^{r,p}_{\pi} \!  \mathcal{M}$ with $q$ Grassmann generators has the tensor structure
\begin{equation*}
\ytableausetup{aligntableaux=center,smalltableaux,boxsize=1.2em} 
	 \ytableaushort{{\sss a_1} {\sss b_1} {\none[\sss \cdots ]} {\sss b_{p\!-\!1}},{\sss a_2},{\none[\sss\vdots]}, {\sss \sss a_{q\!+\!1}} } \ \ ,
\end{equation*}
which matches that of any $\hat{d}$-closed $f$.
\end{proof}
\begin{corollary} \label{corsol}
Let $f \in\Omega^{r,p}_{\pi} \!  \mathcal{M}$ and $g \in \Omega^{r,p+1}_{\pi} \!  \mathcal{M}$ be $\hat{d}$-closed. Then any solution to equation $$\hat{d} f= g$$ has the form $f=\hat{d}^{-1}g + \hat{d}h$ where $\hat{d}^{-1}$ is a partial inverse of $\hat{d}$, and $h \in \Omega^{r,p-1}_{\pi} \!  \mathcal{M}$ is an arbitrary form. 
\end{corollary}
\begin{proof}
Existence of partial inverses follows because $\hat d$ is a fiber-wise operator, 
so we are dealing with a finite dimensional, non-zero linear map. 
Moreover $g$ is in the image of $\hat d$ by Lemma~\ref{biglemma}.
Thus
$\hat{d}^{-1}g$ solves the displayed equation. 
Moreover the difference of two solutions
is necessarily $\hat d$-exact. The result now follows by again employing Lemma~\ref{biglemma}.

\end{proof}
\smallskip
\begin{proof}[Proof of Theorem \ref{decomposition}]
Writing $\omega=\omega_0+\hat{\omega}$, the first equation in $(\ref{quartet1})$ implies the body~$\omega_0$ of $\omega$ is $d$-closed. Also $\hat d \omega_0$ necessarily vanishes so  the second equation says
\begin{equation*}
	\hat{d}\hh\hat{\omega}+dA=0 \, .
\end{equation*}
Hence $dA$ is $\hat{d}$-closed, so using Corollary \ref{corsol} we have
\begin{equation*}
	\hat{\omega}=-\hat{d}^{-1}dA \, . 	
\end{equation*}	
(There is no possibility to add a $\hat{d}$-exact piece because $\hat{\omega} \in \Omega^{2,0}_{\pi} \!  \mathcal{M}$.) Now noting that $d\eta$ is $\hat{d}$-closed, we can solve the third equation using Corollary \ref{corsol},
\begin{equation*}
	A=-\hat{d}^{-1}d\eta+\hat{d}\beta\, ,
\end{equation*} 
for some Grassmann-odd $\beta \in \Omega^{1,0}_{\pi} \!  \mathcal{M}$. Finally using Lemma \ref{biglemma}, the fourth equation can be solved as
\begin{equation*}
	\eta=\hat{d}\gamma \, ,
\end{equation*}
where $\gamma \in \Omega^{0,1}_{\pi} \!  \mathcal{M}$ is an arbitrary Grassmann-odd form. Orchestrating the above gives
$$\Omega=\omega_0+(1-\hat{d}^{\hh -1}d)\hh\hat{d}\beta
+\big(1-\hat{d}^{\hh-1}d 
+(\hat d^{\hh-1} d)^2\big)\hh\hat{d}
\gamma\, .$$
The operator $\hat d^{\hh -1} d$ is, respectively,  two and three-step nilpotent when acting on~$\hat d \beta$ and~$\hat d \gamma$. The result displayed in the theorem is a succinct rewriting of the above display based on this latter fact.
\end{proof}

\subsection{Supersymplectic Inner Products}\label{inner}
It is no longer obvious how to employ the simple mechanism in Equation~\nn{silly} where states were expressed as---manifestly positive---squares of functions, because  the algebra of superfunctions is Grassmann. 
One might require positivity of the body of superfields, but this imposes 
no conditions
on higher filtered components.
To resolve this difficulty, we observe the following: Odd generators on a supermanifold  can be viewed as generalized one-forms.
Moreover, on a Riemannian manifold, the Hodge pairing yields an inner product on differential forms. Therefore we aim  to employ the supersymplectic form to mimic the Hodge construction. (Note that a supergravity motivated Hodge construction on supermanifolds is also studied in~\cite{Cast}.)
One further important observation is that there is a direct correspondence between the  Grassmann analog of the Moyal star product and the Hodge construction. Ironically therefore, we will employ  star products (whose origins are quantum mechanical) to describe classical measurement. \\

In fact, even for a pair of compactly  supported functions $f$ and $g$ on ${\mathbb R}^{2n}\ni \xi^i$, with its standard symplectic form $\omega_{\rm std}=\frac{1}{2}\hh \omega_{ij}\hh d\xi^i\wedge d\xi^j $  ($\omega_{ij}$ is a constant matrix and $\omega^{ij}$ is its inverse), and accompanying Moyal star product
$$
f\star g = f \exp \left(\frac{1}{2}\hh  \overset{\leftarrow}\partial_{\xi ^i}\hh {\omega^{ij}\hh \overset{\rightarrow}\partial_{\xi ^j}}\right)g\, ,
$$
one has that \cite{curtright2013concise}
$$
\int_{{\mathbb R}^{2n}} (\omega_{\rm std})^{\wedge n} fg = 
\int_{{\mathbb R}^{2n}} (\omega_{\rm std})^{\wedge n} f\star g\, . 
$$
The right hand side of the above can also be viewed as a quantum mechanical trace of the  operators $\hat f$ and $\hat g$ obtained by the Moyal quantization of the corresponding classical functions (see~\cite{fedosov1994simple}). In summary, our aim is to generalize the right hand side of the above display to symplectic supermanifolds.

\subsubsection{Supersymplectic volume form}
Note that on any super phase-space $(\mathcal{M},\Omega)$, the symplectic form $\Omega$ is non-degenerate and defines a natural section $\operatorname{Ber}(\Omega)$ of the Berez\-inian~$\text{Ber}\mathcal{M}$. On a representative $(\mathcal{E}M,\pi)$, $\operatorname{Ber}(\Omega)$ is locally trivialized as 

\begin{equation}\label{berez}
dx^1\cdots dx^{2n} \frac{\partial}{\partial \theta^m} \cdots \frac{\partial}{\partial \theta^1} \circ \text{sdet}^{1/2}(\Omega) \in\Gamma( \text{Ber}_{\pi} \mathcal{M})\, , \end{equation}
where  
\begin{equation*}
	\operatorname{sdet}(\Omega) = \frac{\operatorname{det}(\omega-A\eta^{-1}A^T)}{\operatorname{det}(\eta)}=\frac{\operatorname{det}(\omega)}{\operatorname{det}(\eta-A^T\omega^{-1}A)} 
	   \, , 
\end{equation*}
and $\Omega=\omega+A+\eta$ with  $\omega \in \Omega^{2,0}_{\pi} \!  \mathcal{M}$, $A \in \Omega^{1,1}_{\pi} \!  \mathcal{M}$ and $\eta \in \Omega^{0,2}_{\pi} \!  \mathcal{M}$.
Note that the integral
$
	\int_\mathcal{M} \operatorname{Ber}(\Omega) 
$
 vanishes for the standard supersymplectic form on any patch in $\mathbb{R}^{n|m}$, so clearly must be modified to obtain a notion of supersymplectic volumes.
 For that we 
 introduce a canonical  superfield~$\Theta \in \mathcal{A}_m$ on a symplectic supermanifold~$(\mathcal{M},\Omega)$ that  plays the {\it r\^ole} of a vertical volume form.

Our construction relies on a set of contraction operations between arbitrary tensor powers of bilinear forms and suitable sets of vectors on $\mathcal{M}$. Let $B\in \Gamma({\otimes^2}T^*\mathcal{M})$ be any covariant bilinear form on $\mathcal{M}$ and $k$ a  positive integer. Then define
$B^k\in \Gamma({\otimes^{2k}}T^*\mathcal{M})$ by
$$
B^k(X_1,\ldots,X_k; Y_1,\ldots,Y_k):=\frac{1}{k!} \hh B(X_1,Y_1) \hh\cdots B(X_k,Y_k) \, , 
$$
where
$(X_1,\cdots,X_k), (Y_1,\cdots,Y_k) \in \Gamma(T\mathcal{M})\times \cdots \times \Gamma(T\mathcal{M})$
(note that $B^k$ differs from $B^{\otimes k}$ only by a permutation of its arguments).
The tensors $B^k$  descend to any representative vector bundle $\mathcal{E}M$
and act on density-valued tensors in the obvious way. 
Observe that, applied to the 
 top polyvector (\ref{toppolyvec}), we then obtain a fiber density on $\mathcal{E}M$ given by
$$
\Omega^m( \epsilon, \epsilon) = \frac{1}{m!} \hh\epsilon^{a_1\cdots a_m} \Omega_{a_1b_1}\cdots\hh\hh \Omega_{a_mb_m} \epsilon^{b_1\cdots b_m} = \det \eta\, .
$$
When fed the Eulerian top vector~\nn{eulertop} and  top polyvector,  $\Omega^m$ produces 
$$
\Omega^m(\epsilon, \X) = \det \eta \ \theta^1\cdots \hh \theta^m \, . 
$$
Hence we have
	\begin{equation}\label{soccer}
	\Theta(\Omega):= \Omega^m(\epsilon, \X) / \Omega^m(\epsilon, \epsilon)^{1/2} = \sqrt{\det \eta} \:\hh \theta^1\cdots\hh \theta^m\, .	
\end{equation}
It can be explicitly checked  using Equation \nn{superdiff} that the above is a superfield (defined independently of the choice of splitting), indeed
$$\Theta\in {\mathcal A}_m\, .$$
We dub the canonically defined superfield $\Theta$ the {\it vertical volume field}.
We can now define the {\it symplectic supervolume} of $({\mathcal M},\Omega)$ by
$$
{\rm Vol}_{\Omega}({\mathcal M}):=\int_{\mathcal M} \operatorname{Ber}(\Omega)\hh  \Theta(\Omega)\, ,
$$
when this integral exists.
For a patch ${\mathcal U}$ in $({\mathbb R}^{n|m},\Omega_{\rm std})$ over a compact set $U\subset {\mathbb R}^n$, the above returns
the standard Euclidean volume of $U$. Note that using~\eqref{berez} and~\eqref{soccer}, for any superfield $F$, in a choice of splitting $\Omega=\omega+A+\eta$, one has
\begin{equation}\label{reduced}
	\int_\mathcal{M} \text{Ber}(\Omega) \Theta(\Omega) F=\int_M \omega_0^{\wedge \frac{\text{dim} M}{2}} F_0\, .
\end{equation}
We also require the following invariance property of the above measure with respect to supersymplectomorphisms.
\begin{lemma}\label{preserve}
	Let $(\mathcal{M},\Omega)$ be a super phase-space. Suppose $X\in \Gamma(T\mathcal{M})$ obeys $\mathcal{L}_X\Omega=0$. Then, for any compactly supported $F \in \mathcal{A}$ 
	\begin{equation*}
	 \int_{\mathcal M} \operatorname{Ber}(\Omega)\hh  \Theta(\Omega)\, {\mathcal L}_X F=0\, .
	\end{equation*}
\end{lemma}
\begin{proof}
	Making a choice of splitting, we decompose $\Omega=\omega+A+\eta$ in the usual way. Similarly, $X=x+\chi$ where $x$ is a derivation along the base. The body of $\mathcal{L}_X \Omega$ must vanish separately and is given by  
	$$\mathcal{L}_{x_0} \omega_0\, ,$$ 
	where $\omega_0\in \Omega^{2}M$ and $x_0\in \Gamma(TM)$.
	 But using~\eqref{reduced}, we deduce that
	  \begin{equation*}
	 \int_{\mathcal M} \operatorname{Ber}(\Omega)\hh  \Theta(\Omega)\, {\mathcal L}_X F=\int_{M} \omega_0^{\wedge\frac{\dim M}{2}} \hh  \, {\mathcal L}_{x_0} F_0\, ,
	\end{equation*}
	where $F_0$ is the body of $F$. The result now follows from the compact support of $F_0$.
\end{proof}

Finally note that the above construction, as well as the corresponding analog of Lemma~\ref{preserve}, applies to the case where $\Omega$ is a hermitean supersymplectic form. 


\subsubsection{Super stars}
The supersymplectic generalization of the Moyal star product can be employed to quantize 	
 the algebra of  superfunctions\cite{
 berezin1980feynman, fairlie1989infinite, fradkin1991quantization, hirshfeld2002deformation, hirshfeld2004cliffordization, galaviz2008weyl}.
 Indeed given the data of a suitable  superaffine connection, it is possible to define a super Moyal-$\filledstar$  product~\cite{Bor,Hir1,Hir2}.
However what we  require  is a fiberwise inner product in order to extract classical probabilities from superfunctions.
A simple example is illustrative.

\begin{example}\label{fermistates}
Consider the super manifold ${\mathbb R}^{0|2}$ equipped with a hermitean supersymplectic form
$$
\Omega = 
i\big(d\theta^1 \wedge d \theta^1+d\theta^2\wedge d\theta^2)
$$
and a
 hermitean superfield
$$\Lambda \mathbb{R}^2\ni
 f_0 + f_1 \theta^1 + f_2  \theta^2 -i f_{12} \theta^1 \theta^2=F\, ,$$
 where
  $f_0,f_{1},f_2,f_{12}\in {\mathbb R}$.
We now define the map/representation
$$
\sigma(F)=\begin{pmatrix}
f_0 + f_{12}& \bar f \\ f & f_0- f_{12}
\end{pmatrix}= \sigma(F)^\dagger\, ,
$$
where $f:=f_1+if_2$ and the above dagger  is the standard involution on complex $2\times 2$ matrices.
Given $F,G\in {\mathcal A}$ we can then define a star product by demanding
$$
\sigma(F \filmstar G):=\sigma(F)\, \sigma(G) \, ,
$$
in which case we have that
$$
 \filmstar :=  \mathbb{1}\, +
 \stackrel\leftarrow
 {\frac\partial{\partial\theta^1}} \hh
  \stackrel\rightarrow
 {\frac\partial{\partial\theta^1}}
\, +\,  \stackrel\leftarrow
 {\frac\partial{\partial\theta^2}} \hh
   \stackrel\rightarrow
 {\frac\partial{\partial\theta^2}}
\, +\,  \stackrel\leftarrow
 {\frac\partial{\partial\theta^2}} \stackrel\leftarrow
 {\frac\partial{\partial\theta^1}} \hh
\stackrel\rightarrow{\frac\partial{\partial\theta^1}}
  \stackrel\rightarrow {\frac\partial{\partial\theta^2}}\hh=\hh\exp\Big( \frac1i\stackrel\leftarrow{\partial}_A\Omega^{AB} 
  \stackrel\rightarrow \partial_B\!\Big)\, .
$$
Note that $\theta^1\filmstar \theta^1=1=\theta^2\filmstar \theta^2$, $\theta^1\filmstar \theta^2=\theta^1\theta^2=-\theta^2\filmstar \theta^1$. Indeed, this star product encodes the multiplication rule for matrices, and thus is necessarily associative.
The (super) Moyal product for general hermitean superfields $F,G$ yields
\begin{equation*}
	F \filledstar G = f_0g_0 + f_1g_1+f_2g_2+ f_{12}g_{12}  + O(\theta^1\!,\theta^2)\,.
\end{equation*}
In turn 
\begin{align*}
\int_{\mathcal{M}} \text{Ber}(\Omega)\hh \Theta(\Omega)\hh F\filmstar G&= \partial_{ \theta^2}\partial_{\theta^1}\big( \theta^1\theta^2 \big(
f_0g_0 + f_1g_1+f_2g_2+ f_{12}g_{12}  + O(\theta^1\!,\theta^2)
\big)\\ &=
f_0g_0 + f_1g_1+f_2g_2+ f_{12}g_{12} =:( F,G)_\Omega \,. \hspace{2cm}
\end{align*} 
The above is a {\it bona fide} inner product.
Hence, both obviously and importantly for constructing states,
\begin{equation*}
\int_{\mathcal{M}} \text{Ber}(\Omega)\hh \Theta(\Omega)\hh F\filmstar F= f_0^2 + |f|^2+ f_{12}^2\geq 0\,. \hspace{2cm}
\end{equation*}  
%
%
\hfill$\blacksquare$
\end{example}\medskip
The above example demonstrates that, on a purely fermionic vector space equipped with a non-degenerate, hermitean,  symplectic form $\Omega$, the Moyal product of hermitean superfields integrated against the {\it supersymplectic volume density}  $\operatorname{Ber}(\Omega)\Theta(\Omega)$ yields an inner product. This is precisely the Hodge pairing (coming from the symmetric bilinear form $-i\hh\Omega|_{\Lambda\mathbb{R}^m}$) of $F$ and $G$ viewed as generalized differential forms. This construction can be applied \textit{mutatis mutandis} to any purely fermionic vector space ${\mathbb R}^{0|m}$; see Section~\ref{probably}. Of course instantaneous super laboratories (see Section~\ref{superlabs}) are more general supermanifolds than purely fermi vector spaces. To generalize the above construction we need the following technical  lemma.  \begin{lemma} \label{invariantlemma}
	Let $(\mathcal{M},\Omega)$ be a super phasespace, $U$ and $V$ be supervectors, and $(\mathcal{E}M,\pi)$ be a representative vector bundle where $\Omega=\omega+A+\eta$ with $\omega \in \Omega^{2,0}_{\pi} \!  \mathcal{M}$, $A \in \Omega^{1,1}_{\pi} \!  \mathcal{M}$ and $\eta \in \Omega^{0,2}_{\pi} \!  \mathcal{M}$. Then the pairing 
	\begin{equation*}
	\Theta(\Omega) \hh \eta(U,V)  \in \mathcal{A}
	\end{equation*}
defines a   superfunction  on ${\mathcal M}$.
\end{lemma}
\begin{proof}
We need to establish independence of the choice of splitting $(\mathcal{E}M,\pi)$. 
On  a trivializing chart $(X^A)=(x^i,\theta^a)$ for $\mathcal{E}M$,  
on which $U=u^i \frac{\partial}{\partial x^i} + \mu^a \frac{\partial}{\partial \theta^a}$ and 
$V=v^i \frac{\partial}{\partial x^i} + \nu^a \frac{\partial}{\partial \theta^a}$,
the pairing reads
\begin{equation*}
\sqrt{\det \eta} \hh\hh \theta^1 \cdots \theta^m  \eta_{ab} \mu^a \nu^b\, . 
\end{equation*}
The transformation of vectors and the bilinear form under an arbitrary superdiffeomorphism $X^A \mapsto \tilde{X}^A(X)$ is controlled by the 
super Jacobian 
\begin{equation*}
	J^A{}_B = \frac{\partial X^A}{\partial \tilde{X}^B}\, , \qquad (J^{-1})^A{}_B=\frac{\partial \tilde{X}^A}{\partial X^B}\, , 
\end{equation*}
{\it videlicet},
\begin{equation*}
	\nu^a \mapsto \tilde{\nu}^a = J^a{}_B V^B\, , \qquad \eta_{ab} \mapsto \tilde{\eta}_{ab}= (J^{-1})^C{}_a \Omega_{CD} (J^{-1})^D{}_b\, . 
\end{equation*}
The invariance of $\Theta(\Omega) \hh \eta(U,V)$ follows from that of $\Theta(\Omega)$, noting that $J^i{}_a = O(\theta^1)$, and $J^a{}_b = O(\theta^0)$. In particular
\begin{equation*}
\Theta(\Omega) J^a{}_D (J^{-1})^C{}_a =\Theta(\Omega) \delta_d{}^c\, , 
\end{equation*} 
where $D=(j,d)$ and $C=(i,c)$ (and any terms with $D=j$ or $C=i$ vanish). 
\end{proof}
This lemma brings us one step closer to a vertical Moyal product on super phase-spaces. 
Firstly note that the above result extends directly to hermitean supersymplectic forms.
To proceed further we introduce---by way of a technical assumption---the additional data of a flat, torsion-free superaffine connection $\nabla$ preserving the supersymplectic form, \textit{id est}
\begin{equation*}
\nabla\Omega=0=[\nabla,\nabla \}\, . 
\end{equation*}
In principle it might be that such a connection only exists locally, 
but this suffices for  a local  theory of measurement. 
Moreover, recall from the discussion in Section~\ref{showkeys} that, {\it a l\'a} Koszul, this also gives the data of a splitting, and in turn a canonical involution~$\dagger$ as defined in Section~\ref{Hermitean}.
Indeed, given $\nabla$, it is canonical to choose the splitting given by the  corresponding  Koszul bundle ${\mathcal E}^{\!\nabla} \!M:=(\mathcal{E}M,\pi_\nabla)$. In that case we will often label the splitting choice by $\pi_\nabla$ or simply $\nabla$.

\smallskip
Now consider $\nabla^k F, \nabla^k G \in \Gamma(\otimes^kT^*\mathcal{M})$ where $F,G\in \mathcal{A}$, and $k$ is a positive integer. Using the Koszul splitting $\mathcal{E}M$, we may decompose these as
\begin{equation}\label{vertform}
\nabla^kF=f_k+\varphi_k\, , \qquad \nabla^kG=g_k+\gamma_k\, , 
\end{equation}
where $f_k,g_k \in \Omega^k_{\pi} \!  \mathcal{M} \setminus \Omega^{0,k}_{\pi} \!  \mathcal{M}$ and  $\varphi_k, \gamma_k \in \Omega^{0,k}_{\pi} \!  \mathcal{M}$. Then the proof of Lemma \ref{invariantlemma} can be employed to show that 
\begin{equation}\label{omegapower}
\Theta(\Omega)\hh  \Omega^{-k}(\varphi_k , \gamma_k):= \Theta(\Omega) \left( F  \cev{\nabla}{}^k \hh\eta^{-k}\hh \vec{\nabla}^k G \right) \in \mathcal{A}
\end{equation}
is also a well-defined superfield. Here we used the notation
\begin{equation}\label{etaprod}
F  \cev{\nabla}{}^k \hh\eta^{-k}\hh \vec{\nabla}^k G = 	\left(\nabla_{a_1} \cdots \nabla_{a_k}F \right) \eta^{a_1b_1} \cdots \eta^{a_kb_k}\left( \nabla_{b_1} \cdots \nabla_{b_k}G\right)\, .
\end{equation}
The above establishes the following inner product result.

\begin{lemma}\label{superphasespacestar}
Let $(M,\Omega,\nabla)$ be an $(n,m)$-dimensional super phase-space equipped with a flat, torsion-free superaffine connection preserving $\Omega$. Let $(\mathcal{E}M,\pi_\nabla)$ be the Koszul bundle with $\Omega=\omega+A+\eta$ where $\omega \in \Omega^{2,0}_{\nabla} \!  \mathcal{M}$, $A \in \Omega^{1,1}_{\nabla} \!  \mathcal{M}$ and $\eta \in \Omega^{0,2}_{\nabla} \!  \mathcal{M}$. Then for any~$F,G \in \mathcal{A}$, 
\begin{equation*}
	(F,G)_{\Omega,\nabla} := \sum_{k=0}^m c_k \int_{\mathcal{M}} \text{\rm Ber}(\Omega) \Theta(\Omega) \left( F  \cev{\nabla}{}^k \hh\eta^{-k}\hh \vec{\nabla}^k G \right)
\end{equation*}
is an inner product when $c_k \in \mathbb{R}_{>0}$ and $\eta$ is a bundle metric. 
\end{lemma}

Lemma~\ref{invariantlemma} also implies that the action of the inverse $\eta^{-1}$, seen as a bidifferential operator acting on superfields $F,G \in \mathcal{A}$ defined by 
\begin{equation*}
\Theta(\Omega) \hh F \hh \eta^{-1} \hh G  \in \mathcal{A}\, ,
\end{equation*}
is also invariant.  Moreover, given the data $({\mathcal M},\Omega,\nabla)$, we can define a vertical Moyal star product $\filmstar_{\sss\nabla}$ canonical to this data as well by
\begin{equation}\label{GretaGarbo}
	F\filmstar_{\sss\nabla} G:=\left (\exp \eta^{-1}\right) (F,G)\, ,
\end{equation}
 where the action of $\eta^{-k}$ in the (finite) expansion of $\exp \eta^{-1}$ is defined by decomposing~$F,G$ using~\eqref{vertform} and then employing \eqref{omegapower}.
 This amounts to particular---possibly negative---choices of the coefficients $c_k$ appearing in Lemma~\ref{superphasespacestar}.
While the above display is not an invariantly defined superfield, we do have that
\begin{equation*}
	\Theta(\Omega) F\filmstar_{\sss \nabla} G \in \mathcal{A}\, .
\end{equation*} 
Because $\nabla$ is flat and torsion-free,  the usual associativity property of the star product holds
\begin{equation*}
	(F\filmstar_{\sss \nabla} G)\filmstar_{\sss \nabla} H= F\filmstar_{\sss \nabla} (G \filmstar_{\sss \nabla} H)\, .
\end{equation*}
If in addition to associativity, one desires a star product subject to a positivity condition and an inner product, then the above construction must be extended to hermitean supersymplectic forms.

\subsection{Supersymplectic Probabilities}\label{probably}
Motivated by the silly Equation~\nn{silly},
we wish to define probability states as ``star-squares''. More precisely, 
on a super phase-space $({\mathcal M},\Omega,\nabla)$ 
equipped with a flat, torsion-free superaffine connection preserving $\Omega$,
we can construct a probability cone 

\begin{equation}
\label{Mcone}
{\mathcal C}
=\Big\{
\Psi\in 
\Gamma({\mathcal E}^{\!\nabla}\!M)
 \, \Big|\,  \Psi= \Phi\star_{\sss \nabla}\Phi \mbox{ and } \Phi \in {\mathcal A}\Big\}\, .
\end{equation}
To explain this assertion, we continue  our
running Example~\ref{fermistates}.
\begin{example}\label{conifer}
Let the hermitean superfield $F$ be given as in Example~\ref{fermistates}.
Then 
\begin{equation}\label{Iamequationized}
\frac{F\filmstar F}{||F||^2}=\mu_0+\mu_1 \theta^1+\mu_2\theta^2-i\mu_{12}\theta^{1}\theta^2 \, ,\end{equation}
where 
$$||F||^2:=f_0^2+f_1^2+f_2^2+f_{12}^2\,\,
\mbox{ and }\mu_0=1\, ,\:\mu_1=\frac{2f_0f_1}{||F||^2}\, ,\:\mu_2=\frac{2f_0f_2}{||F||^2}\, ,\:\mu_{12}=\frac{2f_0f_{12}}{||F||^2}\, .$$
Employing standard spherical coordinates $(r,\theta,\varphi,\chi)$ on ${\mathbb R}^4\ni(f_0,f_1,f_2,f_{12})$
yields
$$
\mu_1=\sin(2\theta)\cos\varphi\, ,\quad
\mu_2=\sin(2\theta)\sin\varphi\cos\chi\, ,\quad
\mu_3=\sin(2\theta)\sin\varphi\sin\chi\, ,
$$
and hence $\mu_1^2+\mu^2_2+\mu_{12}^2 = \sin^2 (2\theta)$ with $\theta\in [0,\pi]$. It follows that the vector $\vec \mu = (\mu_1,\mu_2,\mu_{12})$ lies in the unit ball in ${\mathbb R}^3$ and labels rays $[\mu_0:\mu_1:\mu_2:\mu_{12}]=[1:\vec \mu]\in {\mathbb R}^4/{\mathbb R}_{>0}$. 
This describes the probability cone. 

Now consider an observable  $$X=x_0+x_1 \theta^1 +x_2 \theta^2 - i x_{12} \theta^1 \theta^2=X^\dagger\, ,$$
given by any hermitean superfield. Then
(in this example we suppress the Koszul connection since implicitly we are using $\nabla=d$ throughout) the corresponding expectation with respect to the state $\Psi=F\filmstar F$  is 
\begin{equation}\label{langX}
\langle X\dangle_{\Psi,\Omega}=\frac{\int_{\mathcal{M}} \text{Ber}(\Omega)\hh \Theta(\Omega)\hh F\filmstar X\filmstar F}{\int_{\mathcal{M}} \text{Ber}(\Omega)\hh \Theta(\Omega)\hh F\filmstar F}=
x_0 + \vec \mu\hh \cdot\hh \vec x\in {\mathbb R}\, .
\end{equation}
Here $\vec x = (x_1,x_2,x_{12})$ and $\vec \mu$ is as above. Finally, one might ask how to express this result in terms of probabilities taking values in $[0,1]$ and summing to unity (after all, the components of $\vec\mu$  are elements of $[-1,1]$). 
For this, special observables/random variables, typically called indicator functions, can be employed as follows.

First, from Equation~\nn{Iamequationized} we see that any normalized state takes the form 
$$
\Psi=1+ \nu_1\theta^1 + \nu_2 \theta^2 -i \nu_{12} \theta^1 \theta^2\, ,
$$
where $\vec \nu$ is in the unit ball. This defines a convex  cone in ${\mathbb R}^4$ over a three-ball. 
Such a cone is is not finitely generated. Nonetheless we may define a four-state discrete system 
 whose states are subset of the set of all 
probabilities for four independent events.
There are of course many ways to do this, for example, we can embed the three-ball in a regular tetrahedron which in turn defines a larger convex cone with four generators. For that consider the following four hermitean superfields
\begin{align*}
\psi_1&=1+\sqrt 8 \hh \theta^1+ i \theta^1\theta^2\, ,
\qquad\qquad\quad \:
\psi_2=1-\sqrt2\theta^1+\sqrt 6 \theta^2
+ i \theta^1\theta^2\, ,\\
\psi_3&=1-\sqrt2\theta^1-\sqrt 6 \theta^2
+ i \theta^1\theta^2\, ,\qquad
\psi_4 =1-3i\theta_1\theta_2\, .
\end{align*}
Then it can be checked that any normalized state $\Psi$ can be expressed as
$$
\Psi = p_1 \psi_1+p_2 \psi_2+p_3 \psi_3+p_4 \psi_4\, ,
$$
so long as
\begin{equation}\label{constraint}
p_1+p_2+p_3+p_4=1\: \mbox{ and }\:
6 ({p_1} -{p_2})^{2}+2 (p_1+p_2-2p_4 )^{2}+ (4p_3
 -1)^{2}\leq 1\, .
\end{equation}
The two relations above constrain each of the four {\it probabilities} $p_i$ ($i=1,\ldots,4$) to be separately in the interval $[0,\frac12]$. (Note that it is not difficult to write down an explicit map between the four probabilities $p_i$,  corresponding to this particular choice of random events, and the coefficients $(f_0,f_1,f_2,f_{12})$ of the underlying state function $F$.)

To call $p_i$ probabilities we ought identify the corresponding events. 
The four observables $X_i:=\frac{\psi_i+2}{12}$ 
obey 
$(\psi_i,X_j)=
\delta_{ij}$ with respect to the inner product introduced below in Equation~\nn{fiberber}. More importantly they resolve unity, meaning
$$
X_1+X_2+X_3+X_4=1\, ,
$$
and have expectations
$$
\langle X_i \rangle_\Psi= p_i\, .
$$
In this sense they are indicator functions.\hfill$\blacksquare$
\end{example}

Next we construct the Moyal product for the super manifold ${\mathbb R}^{0|m}$ which will later appear as fibers for more general super phase-spaces.
The construction of this product follows the previous example but now relies on a Clifford algebra (see~\cite{woit2017quantum}). The latter is generated by $\{\gamma^a|a=1,\ldots,m\}$  where 
$$
\{ \gamma^a,\gamma^b\} = 2 \delta^{ab} \operatorname{Id}\, .
$$
For our purposes it suffices to view the generators $\gamma^a$ as a set of $2^{\lfloor \frac m2\rfloor}\times 2^{\lfloor \frac m2\rfloor}$, hermitean, trace-free, matrices  and $\operatorname{Id}$ is represented by the identity matrix (such a representation certainly exists, see for example~\cite{VP}). 
Now we need to define an invertible  map $\sigma$ from superfields to the Clifford algebra.
Consider  the hermitean superfield
$$
F= f_0 + f_a \theta^a +\frac{ 1}{2!i} \hh f_{ab} \theta^a \theta^b 
+\frac{ 1}{3!i}\hh  f_{abc} \theta^a \theta^b \theta^c
+\frac{ 1}{4!}\hh  f_{abcd} \theta^a \theta^b \theta^c\theta^d + \frac{ 1}{5!}\hh  f_{abcde} \theta^a \theta^b \theta^c \theta^d \theta^e+\cdots \, .
$$
Then we set
\begin{equation}\label{sigCliff}
\sigma(1)=\operatorname{Id}\, ,\quad
\sigma(\theta^a) = \gamma^a\, ,\quad
\sigma(\theta^a \theta^b) \stackrel{a\neq b} = \gamma^a \gamma^b\, ,\quad
\sigma(\theta^a \theta^b\theta^c) \stackrel{a\neq b\neq c} = \gamma^a \gamma^b\gamma^c\, ,\: 
 \ldots\, ,
\end{equation}
and define $\sigma(F)$ by linearity. In turn we obtain a star product on superfields $F$ and $G$ by setting
$$
\sigma(F\filmstar G) := \sigma(F) \sigma(G)\, ,
$$
where the right hand side is matrix multiplication.
Note that 
$$
\theta^a \filmstar \theta^b + \theta^b \filmstar \theta^a= 2\delta^{ab}\, .
$$
Moreover, if $F$ is a hermitean superfield, one has that
$$
\sigma(F)^\dagger = \sigma(F)\, ,
$$
where $\dagger$ here is  complex conjugation composed with matrix transposition.
Hence the matrix trace, $\operatorname{tr} \sigma(F)^2$, is a non-negative real number.
The following result shows that this star product can be used to construct a probability cone on the space of hermitean superfields on  ${\mathbb R}^{0|m}$.

\begin{proposition} 
Let ${ C}$ be the space of superfields $\Psi$ on  ${\mathbb R}^{0|m}$ such that
$$
\Psi = F\star F\, ,
$$
where $F$ is any hermitean superfield. With respect to the  norm defined by 
$$||F||^2:=\frac{1}{2^{\lfloor \frac m2\rfloor}} \operatorname{tr} \sigma(F)^2\, $$ the set ${ C}$ is a probability cone.
\end{proposition}

\begin{proof}
The star-squared property is clearly preserved by scalar multiplication by a positive real number, so
to show that ${C} $ is a probability cone we need to establish the convexity property
	$$C+{C} \subseteq {C}\, .$$ 
	This means the sum of a pair  $\filmstar$-squares must itself be a $\filmstar$-square. Since $\filmstar$  is defined by the matrix product of hermitean matrices, it suffices to show that the sum of squares of two hermitan matrices
 is itself the square of another hermitean matrix. 
The result then follows by standard linear algebra.	
\end{proof}

We can also construct an inner product  between hermitean superfields
$$
(F,G):=\frac{1}{D} \operatorname{tr}
\big( \sigma(F) \sigma(G)\big) = \frac{1}{D} \operatorname{tr} \sigma(F\filmstar G)=
f_0 g_0+
\sum_{k=1}^m \big(\prod_{i=1}^k\sum_{a_i=1}^m\big) \frac1{k!}\hh f_{a_1\cdots a_k} 
g_{a_1\cdots a_k}\,  , 
$$
where $D:=2^{\lfloor \frac m2\rfloor}$. This normalization ensures that the unit superfield $F=1$ has unit norm.
The above construction  relies on a
hermitean symplectic form (or fiberwise metric)
$$
\Omega = i \delta_{ab} d\theta^a  \wedge d \theta^b\, .
$$ 
The star product can be expressed as $\filmstar=\exp\Big( \frac1i\stackrel\leftarrow{\partial}_A\Omega^{AB} 
  \stackrel\rightarrow \partial_B\!\Big)$. Moreover it is easy to check (this follows almost directly from the computation given in Example~\ref{fermistates}) that the above inner product can be obtained from the superintegral
\begin{equation}\label{fiberber}
( F,G)_\Omega=\int_{{\mathbb R}^{0|m}} \text{Ber}(\Omega)\hh \Theta(\Omega)\hh F\filmstar G\, .
\end{equation}

\medskip

It  is interesting to  develop a bound analogous to that given for the vector~$\vec \mu$ in 
Example~\ref{Iamequationized}. 
A good reference for general theory of convex cones applied to probabilities is~\cite{Schneider}.

\begin{proposition}
Let $F$ be a hermitean superfield on    ${\mathbb R}^{0|m}$. Then 
$$
{\mathscr U}
:=\frac{F\filmstar F}{||F||^2}-1\, ,
$$
where $||F||^2 := (F,F)_\Omega$, obeys the bound
$$||{\mathscr U}||^2\leq (
2^{\sss\frac12\lfloor \frac m2\rfloor}
+1)^2- \frac 15\, .
$$
\end{proposition}

\begin{proof}
Note that ${\mathscr U}$ is the soul  of $\check \Psi
:=\frac{F\filmstar F}{||F||^2}$ and plays the {\it r\^ole} of $\vec \mu$ in our earlier example.
Decomposing the normalized state function
$$
\frac{F}{||F||}=\check f_0+{\mathscr F}
$$
into its body and soul, we have that
$
\check f_0^2 + ||{\mathscr F}||^2=1
$
because $(1,\theta^{a_1}\cdots \theta^{a_k})_\Omega=0$ (the trace of any product of distinct  gamma matrices vanishes).
Thus we may write 
$$
\check f_0= \cos \theta\, \mbox{ and }
{\mathscr F}= \sin  \theta \check
{\mathscr F}
$$
for some angle $\theta$, where $||\check
{\mathscr F}
||=1$.
In turn
$$
\check \Psi = 
1
-\sin^2\theta + \sin (2\theta) \check
{\mathscr F}+ \sin^2 \theta \check
{\mathscr F}^{\filmstar 2}\, .
$$
Hence
$$
{\mathscr U}\filmstar {\mathscr U} =
 \sin^2(2\theta) \check
{\mathscr F}^{\filmstar 2} 
-2 \sin(2\theta) \sin^2\theta 
(\check
{\mathscr F}- \check
{\mathscr F}^{\filmstar 3})
 + \sin^4 \theta\hh  (1- \check
{\mathscr F}^{\filmstar 2})^{\filmstar 2}\, ,
$$
so
\begin{align*}
||
{\mathscr U}||^2 &=
\sin^4\theta
+(\sin^2 (2\theta)-2\sin^4 \theta) ||\check
{\mathscr F}||^2 
+2 \sin(2\theta) \sin^2\theta \, \big(\check{\mathscr F}^{\filmstar 3}\big)_0
+\sin^4 \theta\hh
||\check
{\mathscr F}^{\filmstar 2}||^2 
\\
&=
\sin^2 (2\theta)+\sin^4 \theta\hh  (
 \hh
||\check
{\mathscr F}^{\filmstar 2}||^2-1 )
+2 \sin(2\theta) \sin^2\theta \, \big(\check{\mathscr F}^{\filmstar 3}\big)_0
\, .
\end{align*}
Since $\check{\mathscr F}$ is represented by a hermitean matrix, 
$$||\check{\mathscr F}^{\filmstar 2}||^2
=\frac{1}{D} \operatorname{tr} \big(\sigma(\check {\mathscr F})^4\big)
\leq
\frac{1}{D} \Big(\operatorname{tr} \big(\sigma(\check {\mathscr F})^2\big)\Big)^2=D\, $$
and
$$
\Big| \big(\check{\mathscr F}^{\filmstar 3}\big)_0\Big|
=\Big|\frac{1}{D} \operatorname{tr} \big(\sigma(\check {\mathscr F})^3\big)\Big|
\leq
\sqrt D\, .$$
 Hence
$$||{\mathscr U}||^2\leq \big| D \sin^4 (\theta)
+\sin^2(2\theta)-\sin^4\theta
\big| +2 \sqrt D 
\leq (\sqrt D+1)^2- \frac 15\, .
$$
\end{proof}

%
%
%


We have just shown how to produce a probability cone along fibers but
still need to show that Equation~\nn{Mcone} actually defines a probability cone, so we now
return to the data of 
 a super phase-space $({\mathcal M},\Omega,\nabla)$ 
equipped with a flat, torsion-free superaffine connection preserving $\Omega$. 
%
The Koszul bundle ${\mathcal E}^{\nabla\!}M=\Lambda {\mathcal V}M$ of $\nabla$ is equipped with a fiberwise inner product coming from the bundle metric $\eta$, for which we can consider the corresponding bundle of frames. This is an $O(m)$ principal bundle over the base $M$. In turn there is (at least locally) an associated Clifford bundle
${\mathcal C}M$  over $M$ with {\it Clifford map}~$\gamma$
$$
\Gamma({\mathcal V}M)\ni v\longmapsto \gamma(v)\in \Gamma({\mathcal C}M)\, ,
$$
where
$$
\gamma(v) \gamma(u) + \gamma(u)\gamma(v) = 2 \eta(v,u)\operatorname{Id}\, ,
$$
for any pair of sections $v$ and $u$. Here $\operatorname{Id}$ is the unit section of ${\mathcal C}M$. 
 The Clifford map extends to  act on 
tensor powers of vertical vectors according to
$$
\gamma(v_1 \otimes \cdots \otimes v_k)= \gamma(v_1)\cdots \gamma(v_k) \in \Gamma({\mathcal C}M)\, .
$$
In turn this defines $
\gamma(p) \in \Gamma({\mathcal C M})
$ for any
 section $p\in \Gamma(\Lambda^k{\mathcal V}M)$.
 Indeed any section of~${\mathcal C}M$
 can uniquely decomposed as a sum $\sum_{k=0}^m \gamma(p_k)$ where $p_k\in \Gamma(\Lambda^k{\mathcal V}M)$. Note that $\gamma(f):=f\operatorname{Id}$ when $f\in C^\infty M=\Gamma(\Lambda^0M)$.

\smallskip

As discussed earlier, because $\nabla$ is flat and torsion-free, given a superfield $F$
we may employ the inverse of $\eta \in \Omega^{0,2}_{\nabla} \!  \mathcal{M}$ to  make a section 
$\eta^{-k} (\nabla^kF)\:\in \Gamma(\Lambda^k{\mathcal V}M)$. (Note that this only uses the vertical part of $\nabla^k F$.) In turn, fed to~$\gamma$, we get  a section 
$$
\gamma(\nabla^k F):=\gamma\big(\eta^{-k} (\nabla^k F)\big) \in
\Gamma({\mathcal C}M)\, . 
$$

We may now define an invertible map from superfields $F$ to Clifford bundle sections
$$
\sigma(F)=\gamma\big (  (\exp \eta^{-1} \nabla F )_0 \big)  
=\operatorname{Id} F_0 + \gamma(\nabla F) + \frac1{2!}\gamma(\nabla^2 F) + \cdots + \frac1{m!} \gamma(\nabla^m F)\, .
$$
Thus  Clifford multiplication can be used to 
define a vertical star product 
via
$$
\sigma(F\filmstar_{\nabla} G)=\sigma(F) \sigma(G)\, .
$$
It is not difficult to check  (by applying our
earlier discussion of ${\mathbb R}^{0|m}$ along fibers)
that
the Clifford definition of the star product agrees with the one given earlier, {\it id est},
$$
F\filmstar_{\nabla} G=\left (\exp \eta^{-1}\right) (F,G)\, .
$$
The Berezinian integral then gives us a
manifestly positive inner product between 
hermitean superfields $F$ and $G$ on ${\mathcal M}$,
\begin{equation*}
(F,G)_{\Omega,\nabla}:=
	\int_\mathcal{M} \text{Ber}(\Omega) \Theta(\Omega) F\filmstar_{\nabla} G=\int_M \omega_0^{\wedge \frac{\text{dim} M}{2}}
	\Big(
	F_0 G_0
	+
	\sum_{k=1}^m\big(
		\eta^{-k}(\nabla^k F,\nabla^k G)\big)_0
	\Big)\, ,
\end{equation*}
where $\big(\eta^{-k}(\nabla^k F,\nabla^k G)\big)_0$ is  the function on $M$ 
defined by the body of Equation~\nn{etaprod}.

\smallskip

The probability cone ${\mathcal C}$ of Equation~\nn{Mcone} is now defined. That this is indeed a probability cone in fact follows directly from the fiberwise examples given above.
In the language of Definition~\ref{mylozengeisblack}, 
the vector bundle ${\mathcal V}Z$ is ${\mathcal E}^{\sss \nabla} M$. The vector space $V$ is the space~$\Lambda {\mathbb R}^{m}$ while $C$ is the space of elements of the form $\sigma^{-1}\big(\sigma(F)^2\big)$ where $F\in \Lambda {\mathbb R}^{m}$. Pre-states are superfunctions, {\it id est} sections $\Psi$ of ${\mathcal E}^{\sss \nabla} M$ expressible as a $\filmstar_{\sss \nabla}$-square $\Phi\filmstar_{\sss \nabla} \Phi$. When $\int_{\mathcal M} \operatorname{Ber}(\Omega) \Theta(\Omega)
\Phi
\filmstar_{\sss \nabla} \Phi
<\infty$,  the pre-state $\Psi$ is a {\it bona fide} state.

To summarize, on a super phase-space $(M,\Omega)$ equipped with a flat, torsion-free connection~$\nabla$, the expectation of an observable given by a hermitean superfield $X$ with respect to  
a state $\Psi\in {\mathcal C}$ defined by the state function $\Phi$
 is given by
$$
\langle
X
\dangle_\Psi:=\frac{\int_{\mathcal M} \operatorname{Ber}(\Omega) \Theta(\Omega)\, \Phi \filmstar_\nabla \!X\filmstar_\nabla \!\Phi
}{\int_{\mathcal M}\operatorname{Ber}(\Omega) \Theta(\Omega) \, \Phi \filmstar_\nabla  \Phi}\, .
$$
We are now ready to describe dynamics.

\section{Discrete Dynamical  Systems}
\label{sec7}

\subsection{Super Dynamics} Analogous to the bosonic case, dynamics on a symplectic supermanifold shall be 
 generated by Grassmann even super vectors. 
 On a symplectic supermanifold $(\mathcal{M},\Omega)$, a \textit{symplectic super vector} is any super vector $\Rho$ that satisfies
\begin{equation}
	\mathcal{L}_{\Rho}\Omega =0\, .  \label{superrho}		
\end{equation}
In analogy with the bosonic case, we make the following definition. 

\begin{definition}
	A \textit{dynamical super phase-space} is a triple~$(\mathcal{M},\Omega,L)$ where $(\mathcal{M},\Omega)$ is a symplectic supermanifold and $L\subset T\mathcal{M}$ is a Grassmann-even line subbundle whose sections are symplectic super vectors.
\hfill$\blacklozenge$\end{definition}

Given an even superfield $H$ on a 
 super phase-space, one can define the corresponding {\it Hamiltonian super vector field} $X_H$ by $$\flat X_H = d_TH\, .$$ Clearly  $X_H$ is a symplectic super vector. Once again (see Example \ref{oddsympdynamics}), this data can be succinctly  repackaged in terms of  a super phase-spacetime:
\begin{example}
Let $(\mathcal{M},\Omega)$ be a super phase-space such that $\Omega = d_T \Lambda$ and let $H$ be an even superfield. Calling 
$$\mathcal{Z}=\mathcal{M}
\times \stackunder{$\mathbb{R}$}{$\stackrel{\rotatebox{90}{$\scriptscriptstyle \in$}}{t}$} \stackrel{\pi}{\to} \mathcal{M}\, , $$ 
we consider the super 1-form
\begin{equation*}
A= \pi^*\Lambda + H dt \in \Omega^1\mathcal{Z}\, . 
\end{equation*}
The pair $(\mathcal{Z},d_TA)$ defines a super phase-spacetime. Moreover the supervector $$i^t_* \Omega^{-1}(d_TH,\pdot)+\partial _t,$$ where $i^t$ is the inclusion $\mathcal{M} \to \mathcal{M} \times \{t\}$, defines a line bundle whose sections are symplectic supervectors and thus in turn defines a dynamical super phase-space. \hfill$\blacksquare$
\end{example}
Quite generally, the supersymplectic form of a super phase-spacetime $(\mathcal{Z},\Omega)$  encodes a line bundle $L = \operatorname{ker} \Omega$, sections of which determine generators of dynamics on~$\mathcal{Z}$. Note that the proof of Proposition \ref{odd-dynamical} applies directly here as well. When the line bundle $L$ is orientable, there is a notion of future directed super dynamics 

The body  $(Z,\omega_0)$ of $(\mathcal{Z},\Omega)$ is itself a symplectic manifold.  When $\mathcal{Z}$ is a super phase-spacetime it follows, for any $\rho$ subject to $\iota_\rho \omega_0=0$, that 
$(Z,\omega_0,\operatorname{span} \rho)$ 
 is a dynamical phase-space.
The following proposition shows how     generators  of dynamics~$\rho_0$ on $(Z,\omega_0)$ lift to super generators on  $\mathcal{Z}$.
\begin{proposition}\label{p2P}
	Let $(\mathcal{Z},\Omega, L)$ be a super phase-spacetime, and $(\mathcal{E}Z,\pi)$ be a representative vector bundle where $\Omega=\omega+A+\eta$ with $\omega \in \Omega^{2,0}_{\pi} \!  \mathcal{Z}$, $A \in \Omega^{1,1}_{\pi} \!  \mathcal{Z}$ and $\eta \in \Omega^{0,2}_{\pi} \!  \mathcal{Z}$. 
	Suppose that 
	$$
	\omega_0(\rho_0,\pdot)=0
	$$
	for some $0\neq \rho_0\in  \Gamma(TZ)$.  
	Assuming $\eta$
is	 
invertible, the general solution of Equation~\nn{superrho} is then
given  
by:
	\begin{equation*}
		\Rho= \rho+\eta^{-1}A(\rho ,\pdot)\, , \quad \rho=
		\Big(\frac{\operatorname{Id}}{\operatorname{Id}+\omega_0^{-1}\circ\hat \omega}\Big)(\rho_0,\pdot\hh)
\, .
	\end{equation*}
	In the above    
	$\hat \omega :=\omega-\omega_0 -A\circ \eta^{-1}\circ A$ and $\omega_0^{-1}$ is any choice of partial inverse
	of $\omega_0$.
%
%
%
\end{proposition}

\begin{proof}
Decomposing 
$$\Rho=\rho + \rho_V\, $$
into its body $\rho$ and soul $\rho_V$, Equation~\nn{superrho} produces
$$
\omega(\rho,\pdot)+ A(\pdot ,\rho_V)=0=A(\rho,\pdot)+\eta(\rho_V)\, .
$$
Solving the second of these equations for $\rho_V$ and back-substituting in the first yields
\begin{equation}\label{sappy}
\bar \omega(\rho)= 0\, ,
\end{equation}
where 
$$
\bar \omega := \omega - A\circ \eta^{-1}\circ A\, .
$$
Upon expanding in the ${\mathbb Z}_{m+1}$-grading determined 
by the splitting $\pi$, Equation~\nn{sappy} becomes a block lower triangular matrix system of equations whose diagonal is $\omega_0$, thus let us write
$$
\hat \omega=\bar \omega-\omega_0 \, ,
$$
where $\hat \omega$ is nilpotent.
Moreover, because $\bar \omega$ is skew, $\bar \omega(\rho_0,\pdot)$ annihilates $\rho_0$, so the covector $\big(\omega_0^{-1}\circ\hat \omega\big)(\rho_0,\pdot)$
is well-defined given a partial inverse (note that the image of $\omega_0$ is the kernel of  $\rho_0$ viewed as a linear map $\Gamma(T^*Z)\to C^\infty Z$). Now suppose $F$ is any even superfield. Then an easy inspection of this matrix system of equations
 shows that the general solution to Equation~\nn{sappy} is
$$
\rho = 
\Big(\frac{\operatorname{Id}}{\operatorname{Id}+\omega_0^{-1}\circ\hat \omega}\Big)({F}\rho_0,\pdot\hh)
\, ,
$$
where $\omega_0^{-1}$ is any choice of partial inverse mapping $\ker \rho_0\subset \Gamma(T^*Z)$ to a choice of subspace of $\Gamma(TZ)\not\ni
\rho_0$.

\end{proof}
\noindent
We have now assembled the super geometric technology required to handle dynamical systems involving both discrete and continuous degrees of freedom.

\subsection{Super Laboratories, States and Measurements} \label{superlabs}

Motivated by our earlier discussion of computer bits and superfields, we now focus on the data of a super phase-spacetime  $({\mathcal Z},\Omega)$.
Our first task is to develop a notion of an instantaneous laboratory along the lines  discussed in Section~\ref{sec2}.
First, if~${\mathcal Z}=(Z,{\mathcal A})$ is a supermanifold and $M\hookrightarrow Z$ is a hypersurface in $Z$, we call the induced supermanifold $(M, {\mathcal A}|_M)$
a {\it superhypersurface}. 
Note that an exponential map $\exp(t v)$ on a supermanifold ${\mathcal Z}$ restricts to superhypersurfaces in the natural way.
In particular, this induces the standard exponential  map on the body, giving $M_t =\exp(t v_0) M$.
We may now make the following definition.


\begin{definition}
An \textit{instantaneous super laboratory} $({\mathcal M},\Omega^*)$ is any
superhypersurface   
 of a super phase-spacetime $({\mathcal Z},\Omega)$ along which the pullback $\Omega^*$ is non-degenerate.
%
%
%
\hfill$\blacklozenge$\end{definition}

Next we  choose a suitable set of observables. For that we employ  the algebra formed by the space of (hermitean) superfunctions ${\mathcal A}$ on a super phase-spacetime ${\mathcal Z}=(Z,{\mathcal A})$. Moreover, given a choice of splitting ${\mathcal E}Z$, these belong to the space of sections of the vector bundle~${\mathcal E}Z$. Given a measure, this aligns with the definition given in Section~\ref{sec2}. 
States are then identified with sections of a convex cone inside the dual bundle $\mathcal{E}Z^*$. This gives a notion of positivity. The space of  (normalizable) superfunctions in the kernel of~$\mathcal{L}_\Rho$, for any~$\Rho$ that obeys~$\Omega(\Rho,\pdot)=0$, is the  fundamental state space, which we will denote by~$\mathcal{A}_\Omega$. 

\begin{lemma}
	Let $(\mathcal{Z},\Omega,\nabla)$ be a super phase-spacetime equipped with a flat, torsion-free superaffine connection preserving $\Omega$. Moreover suppose that $\Rho\in \ker \Omega$.
	 Then there exists a canonical  $\filmstar$-product on  $\mathcal{A}_\Omega$.
		\begin{proof}
	First note that $\Omega(\Rho,\pdot)=0$ implies 
	$$\Omega(\nabla \Rho,\cdot)=0\,.$$
	Thus $\nabla \Rho$ must be proportional to $\Rho$, or in other words $\nabla$ preserves the kernel of $\Omega$.\\[-2mm]
	
		Now let $F,G$ be state functions. 
		Hence $d_TF$ and $d_TG$ are both in the kernel of $\iota_P$.
		Then observe that $$\Omega(X,\pdot)=d_T F$$
		has a unique solution modulo addition of any vector field in the kernel of $\Omega$. 
 Hence $\Omega^{-1}(d_TF,d_TG)$ is also well-defined.\\[-2mm]
		
		Now consider the (covariant) tensor 
		$$X:=\nabla \cdots \nabla d_TF\, .$$
		We wish to establish $X(\pdot,\ldots,\Rho,\ldots,\pdot)=0 $ for any $\Rho\in \ker \Omega$. Since $\nabla$ is torsion-free and flat, it suffices to show $X(\pdot,\ldots,\pdot,\Rho)=0$. This follows from a simple inductive argument. The base case for $X=\nabla d_T F$ looks like
		$$\iota_{\Rho} \nabla d_T F=[\iota_P,\nabla] d_TF \propto \iota _\Rho d_TF=0 .$$
		The proportionality statement follows from the assumption that $\nabla \Rho $ is in the span of $\Rho$. The induction step follows similarly.
		
	A canonical $\filmstar$-product is given by 
	$$F\filmstar G:=\left (\exp \Omega^{-1}\right) (F,G)$$	
where $\Omega^{-k}(F,G)$ is defined in the obvious way. That the above defines a star product again follows from usual derivation of the standard super Moyal star product  because~$\nabla$ is torsion-free and flat, and thus super-commuting (see for example \cite{bayen1978deformation, lichnerowicz1982deformations, fedosov1994simple}).
%
	\end{proof} 
\end{lemma}

In fact for our  classical measurement theory, we only need a vertical star operation. This is in fact simpler to construct than the super Moyal star given above because, as discussed earlier, in a choice of splitting, $\Omega$ is fiberwise invertible.
\begin{definition}
	Let $(\mathcal{Z},\Omega,\nabla)$ be a super phase-spacetime equipped with a flat, torsion-free superaffine connection preserving $\Omega$ and $F,G$ be superfunctions on $\mathcal{Z}$. Then given a superhypersurface $\Sigma$, we define the inner product of $F$ and $G$ by
 	$$(F,G)_{\Omega,\Sigma}:=\int_\Sigma \text{Ber}(\Omega)\Theta(\Omega)F \filmstar_{\sss \nabla} G\, , $$
	where $F\filmstar_{\sss \nabla}G:=\left (\exp \eta^{-1}\right) (F,G)$
and $\eta^{-k}(F,G)$ is as defined in  Equation~\nn{etaprod}.\noindent	\hfill$\blacklozenge$\end{definition}
\noindent
The above definition is a super phase-spacetime analog of Definition~\ref{superphasespacestar}; its well-defined\-ness (with respect to superdiffeomorphisms) follows from the discussion preceding that definition. Next we study how the above inner product behaves with respect to superdynamics. First we need a technical lemma.


\begin{lemma}
	Let $(\mathcal{Z},\Omega,\nabla)$ be a super phase-spacetime equipped with a flat, torsion-free superaffine connection preserving 
	$\ker \Omega$. Then locally there exists a vector $\Rho \in \ker \Omega$ such that $$\nabla \Rho =0\, .$$
\end{lemma}
\begin{proof}
		Since $\nabla$ preserves $\ker \Omega$, there exists a (everywhere non-vanishing on any contractible patch) vector $\Rho' \in \ker \Omega$ such that ,
		$$\nabla \Rho' =\Upsilon\otimes \Rho', $$
		for some super one-form $\Upsilon$.
		Using that $\nabla$ is torsion-free and flat, we observe
		$$0=(\nabla \wedge \Upsilon)\otimes \Rho'\, ,$$
		where we used $\Upsilon \wedge \Upsilon=0$. Hence $\Upsilon$ is $\nabla$-closed and in turn $d_T \Upsilon=0$ since $\nabla$ is torsion-free. So locally 
		$$\Upsilon=d_T F\, ,$$ for some even superfunction $F$. Now consider $\Rho:=e^{-F} \Rho'$. It follows that
		$$
		\nabla\Rho=(-d_T F+\Upsilon)\otimes \Rho'=0\, .$$ 
		\vspace{-1cm}
	\end{proof}
	\vspace{1cm}

The inner product between superfunctions satisfies a conservation law:\begin{lemma}
\label{evolve}
Let $(\mathcal{Z},\Omega,\nabla)$ be a super phase-spacetime equipped with a flat, torsion-free superaffine connection preserving $\Omega$.
 Moreover let $\mathcal{M}_t$ be a smooth 1-parameter family of superhypersurfaces such that $\mathcal{M}_t:=\exp (t\Rho)(\mathcal{M})$, where $\Rho \in \ker \Omega$ and $\nabla \Rho=0$. Then for any pair of superfields $F,G$ their inner product obeys
  $$\frac{d (F,G)_{\Omega,\mathcal{M}_t}}{dt}\Big \vert_{t=0}=(\mathcal{L}_\Rho F,G)_{\Omega,\mathcal{M}}+(F,\mathcal{L}_\Rho G)_{\Omega,\mathcal{M}}\,.$$
 \end{lemma}

\begin{proof} 
Firstly, recall that given a vector field $v$ and a $2n$-form $\mu$, then 
 $$\frac{d}{dt} \left(\int_{M_t} \mu \right)\Big\vert_{t=0}=\int_M \mathcal{L}_v \mu\,,  $$
where $M_t:=e^{t v}(M) \hookrightarrow Z^{2n+1}$ is a 1-parameter family of hypersurfaces generated by $v$. 

Let us now begin with the right hand side of the claimed conservation law, for which we consider the decomposition of $\Omega=\omega+A+\eta$ corresponding to the Koszul bundle of $\nabla$. Using \eqref{reduced} we have 
$$(\mathcal{L}_\Rho F,G)_{\Omega,\mathcal{M}}+(F,\mathcal{L}_\Rho G)_{\Omega,\mathcal{M}}=\int_M \omega_0^{\wedge n} \left[\mathcal{L}_\Rho F \filmstar_{\sss \nabla} G+F \filmstar_{\sss \nabla}\mathcal{L}_\Rho G\right]_0\,. $$ 
We next expand  $\filmstar_{\sss \nabla}$ order by order; a general term looks like
(see~\nn{GretaGarbo} and~\nn{etaprod}) 
\begin{equation}\label{genterm}
	\eta^{-k}(\mathcal{L}_\Rho\nabla^k F,\nabla^k G)+\eta^{-k}(\nabla^k F,\mathcal{L}_\Rho\nabla^k G)\,,
\end{equation}
where we have used that $\Rho$ is parallel with respect to the torsion-free and flat connection~$\nabla$. Furthermore, $\eta$ and hence $\eta^{-1}$ is preserved up to $\mathcal{O}(\theta^0) $ under flows of $\Rho$. To see this note that,
$\mathcal{L}_\Rho \Omega=0 \Rightarrow \omega_0^{\wedge n} \Theta \mathcal{L}_\Rho \Omega=0 \Rightarrow (\mathcal{L}_{\Rho}\eta)\vert_{\theta=0} =0 $. 
By the Leibniz rule, it now follows that the body of the general term 
 \eqref{genterm} is 
$$\mathcal{L}_{\rho_0}\left[ \eta_0^{-k}\left( (\nabla^k F)_0,(\nabla^k G)_0)\right)\right]\, ,$$ where the vector field $\rho_0$ is the body of $\Rho$.

We now consider the left hand side of the claimed conservation law. Applying Equation
~\eqref{reduced} followed by
 the first display of this proof gives
\begin{equation*}
\frac{d (F,G)_{\Omega,\mathcal{M}_t}}{dt}\Big \vert_{t=0}
=
	\int_M \mathcal{L}_{\rho_0}\left( \omega^{\wedge n}F \filmstar _{\sss \nabla} G\right )_0
	\,. 
\end{equation*}
The proof is completed by 
recalling how the exponential map acts on superhypersurfaces and their base manifolds, and then
using  ${\mathcal L}_{\rho_0} \omega_0=0$ in conjunction with~\nn{GretaGarbo} and~\nn{etaprod}.

%
%
%
%
\end{proof}

The above result gives an evolution rule for expectations of observables with respect to changes of laboratory.
\begin{corollary}\label{evolvingexpectations}
Let $(\mathcal{Z},\Omega,\nabla)$ be a super phase-spacetime equipped with a flat, torsion-free superaffine connection preserving $\Omega$. Moreover let $\mathcal{M}_t$ be a smooth 1-parameter family of superhypersurfaces such that $\mathcal{M}_t:=\exp (t\Rho)(\mathcal{M})$, where $\Rho \in \ker \Omega$ and $\nabla \Rho=0$. Then the evolution of the expectation of an observable $X \in \mathcal{A}Z$ with respect to a state function $\Phi \in \mathcal{A}Z$ is given by 
\begin{equation*}
	\frac{d}{dt} \langle X \rangle_{\Phi, \mathcal{M}_t} = \langle \mathcal{L}_\Rho X \rangle_{\Phi, \mathcal{M}_t} \, .
\end{equation*}
\end{corollary}
\begin{proof}
To see this, note that the state function $\Phi$ obeys the Liouville equation implying that total probabilities are preserved. Hence the only contribution to the expectation can come from the change in the observable itself. 
\end{proof}

\section{Example: Evolution of Discrete Probabilities}

We now demonstrate our key ideas with a two-bit system evolving in time. Our basic data is the supermanifold 
$\mathcal{Z}=\mathbb{R}^{1|2}\ni(t,\theta^a)$ with $a=1,2$.
For the flat, torsion-free affine connection we take
$$\nabla = d_T$$
in the coordinate system $(t, \theta^a)$.
Our dynamics shall be described by the  action functional of~\cite{Berezin:1976eg}
\begin{equation*}
	S 
	= \int \tfrac{d\tau}i \big(\eta_{ab}(t)\theta^a \dot{\theta^b} -\frac 12 H(t)\epsilon_{ab}\theta^a\theta^b \dot{t} \big)\, .
\end{equation*}
Assuming $\eta$ is invertible, this
is the most general super phase-spacetime dynamics on~${\mathcal Z}$.
In the above the dot denotes 
differentiation with respect to the 
worldline parameter $\tau$
and $\epsilon_{ab}$ is the two-dimensional Levi-Civita symbol.
Better said, this action is an integral of a hermitean super one-form 
$$
\xi = i\eta_{ab}(t)\theta^a d\theta^b \color{black}+\color{black}\frac i2H(t)\epsilon_{ab}\theta^a\theta^b dt \, .
$$
Here we used that, for a one-form  $d\tau$ (given by $d_T$ of a bosonic function), one has~$d\tau \hh\dot \theta = d\theta$ (this choice of sign ensures that the variation of the above action agrees with the dynamics of the super symplectic form below).
This gives the symplectic form
$$
\Omega=d_T \xi 
=i\eta_{ab}(t)d\theta^a d\theta^b
-i\big(H(t) \epsilon_{\color{black}ba\color{black}} +\eta'_{ab}(t)\big)\theta^b d\theta^a dt\, .
$$
The data $(\mathbb{R}^{1|2},\Omega)$ defines a super phase-spacetime.
Dynamics is determined by the supervector ({\it conferatur} Proposition~\nn{p2P})  
$$
P=\frac{\partial}{\partial t}+
\theta^a M^b{}_a
\frac{\partial}{\partial \theta^b}\, ,
$$
where
$$
M^b{}_a:= -\frac{1}{2}\eta^{bc}(t)\big(H(t) \epsilon_{\color{black}ac\color{black}} +\eta'_{ca}(t)\big)\, ,
$$
and $\eta^{ab}$ is the inverse of $\eta_{ab}$. Also $\prime$ denotes $t$-differentiation.

Now we consider a state function $\Phi$ given by a hermitean superfield
$$
\Phi = \phi + \phi_a \theta^a - \tfrac i2 \varphi\epsilon_{ab}\theta^a \theta^b\, ,
$$
(where $\epsilon_{12}=1$) subject to 
$$
P\Phi = 0\, .
$$
Here the four functions of $t$ given by $(\phi, \phi_a,\varphi)$ are all real and must obey 
\begin{equation}\label{eoms}
\phi'=0\, ,\quad \phi_a' + M^b{}_a  \phi_b =0\, ,\quad
\varphi' + M^a{}_a \varphi = 0\, .  
\end{equation}
Note that $M^a{}_a =-\frac12 \eta^{ab} \eta'_{ab}$.

We now consider the super laboratory $\Sigma_T$ given by superhypersurface $t=T$.
The pullback of $\Omega$
is then
\begin{equation}\label{pullback}
\Omega^*=i \eta_{ab}(T) d\theta^a d\theta^b\, .
\end{equation}
Before computing the star product, we  solve the Equations~\nn{eoms}. 
For that we impose boundary conditions
for
 $\Phi(T)$ along $\Sigma_T$ given by $\Phi(T)=(\phi^{\rm i}, \phi_a^{\rm i},\varphi^{\rm i})$. Then 
$$
\phi(t) = \phi^{\rm i}\, ,\quad
\phi_a(t) = U(t,T)^b{}_a\, \phi_b^{\rm i}\, ,\quad
\varphi(t)=u(t,T)\varphi(T)\, ,
$$
where 
$
U(t,T)^b{}_a=\Big[\operatorname{P} \exp\big(-\int^t_T M \big)\Big]^b{}_a
$, 
$
u(t,T)=\exp\Big({-\int^t_T M^a{}_a} \Big)
=\exp\Big({\frac12 \int^t_T \eta^{ab}\eta^\prime_{ab}} \Big)
$, and~$\operatorname{P}$ denotes  path ordering.

To compute the vertical star product of state functions we first need the Clifford map~$\gamma$. 
 For that we introduce orthonormal frames $e_a{}^{\bar a}$ for the vertical distribution on~$T^*{\mathcal M}$ determined by the Koszul connection, subject to
$$
\eta_{ab}(t) d\theta^a d\theta^b = \delta_{\bar a \bar b} e_a{}^{\bar a}(t)e_b{}^{\bar b}(t) d\theta^a d\theta^b\, .
$$
We  also  define
$
\theta^{\bar a}(t) := e_a{}^{\bar a} (t)\theta^a
$
so that
$$
\Phi = \phi + \phi_{\bar a} \theta^{\bar a} - \tfrac i2 
\bar \varphi \, \epsilon_{\bar a\bar b}\theta^{\bar a} \theta^{\bar b}\, ,
$$
where 
$
\phi_{\bar a} := e^a{}_{\bar a} \phi_a\, ,\quad
\bar \varphi := (\det{}^{-1} e) \hh \varphi
$,
and $\epsilon_{\bar 1\bar 2}=1$.
We may now
  employ the map $\sigma$ of Example~\ref{fermistates} along each fiber.
This gives an invertible map from hermitean superfunctions to hermitean matrices taking values in functions of $t$, 
$$
\sigma(\Phi)=\begin{pmatrix}
\phi +\bar \varphi & \phi_{\bar 1}- i \phi_{\bar 2}\\
\phi_{\bar 1}+i\phi_{\bar 2}& \phi-\bar\varphi
\end{pmatrix}\, ,
$$
such that the star product is computed via matrix multiplication:
$$
\Phi\filmstar_{\sss \nabla} \Phi = 
\sigma^{-1}\big(\sigma(\Phi)^2\big)
=\phi^2 
+\phi_{\bar 1}^2+\phi_{\bar 2}^2
+\bar \varphi^2 
+ 2 \phi (\phi_{\bar 1} \theta^{\bar 1}
+  \phi_{\bar 2} \theta^{\bar 2})
-i \phi \bar \varphi 
 \epsilon_{\bar a\bar b}\theta^{\bar a} \theta^{\bar b}\, .
$$
Hence, once again,  a hermitean superfield $\Psi$ is a vertical star square of a non-zero state function precisely when its body is positive. As discussed earlier (see Example~\ref{conifer}) this condition is both a necessary and sufficient one.

\subsection{Measurements}

We can now construct expectations of observables. For that let $X=x+x_a\theta^a -\frac i2 \chi \epsilon_{ab} \theta^a \theta^b $ be a hermitean superfield, which we may also expand in the adapted Grassmann coordinates~$\theta^{\bar a}(t)$, 
$$
X=x+ x_{\bar a} \theta^{\bar a} - \tfrac i2 
\bar \chi\, \epsilon_{\bar a\bar b}\theta^{\bar a} \theta^{\bar b}\, .
$$
The pullback in Equation~\nn{pullback}
in terms of the adapted $\theta^{\bar a}$ is
$$
\Omega^* = i \delta_{\bar a\bar b} d\theta^{\bar a}d\theta^{\bar b}\, .
$$
Hence
the expectation of the  observable $X$  in the $t=T$ superlaboratory with respect to a state $\Psi = \Phi\filmstar_{\sss \nabla} \Phi $ is given by (see Equation~\nn{langX})
\begin{equation}\label{expect}
\langle X\dangle_{\Psi,T}=
x(T) +
\frac{2\phi (\phi_{\bar 1}x_{\bar 1}+\phi_{\bar 2}x_{\bar 2}+\bar \varphi\bar \chi)}{\phi^2 
+\phi_{\bar 1}^2+\phi_{\bar 2}^2
+\bar \varphi^2 }\Big|_{t=T}
=x(T) +
\frac{2\phi (\phi_a \eta^{ab} x_b+\det^{-1}\! \eta \, \hh\varphi  \chi)}{\phi^2 
+\phi_a\eta^{ab} \phi_b+\det^{-1}\!\eta\,  \hh\varphi^2 }\Big|_{t=T}
\, .
\end{equation}

To simplify the remainder of our discussion, let us assume $\eta$ and $H$ are $t$-independent. 
It follows that $u(T,t)=1$ and 
\begin{equation}\label{letsmove}
U(T,t)=
\cos\big(\tfrac{(t-T)H}{2\surd\det \eta}\big)\operatorname{Id}
+
\sqrt{\det \eta}\, 
\sin\big(\tfrac{(t-T)H}{2\surd\det \eta}
\big)\hh
 \eta^{-1} \epsilon \, .
\end{equation}
 By construction $U^\top \eta \,\hh U = \eta$ so that total probabilities are conserved in the sense that $(\Phi,\Phi)_{\Omega,\Sigma_t}$ is constant,
 in concordance with  Lemma~\ref{evolve}.
 It is then not difficult to see that
 $$
 \left.
 \frac{d\langle X\dangle_{\Psi,t}}{dt}\right|_{t=T}
 =\langle {\mathcal L}_{P}X\dangle_{\Psi,T}\, ,
 $$
 as predicted by Corollary~\ref{evolvingexpectations}
 (note that the 
 supervector $P$ acts on
  the components $(x,x_a,\chi)$ of $X$ according to 
$(x',(x'+Mx)_a,\chi')$).

Finally we can consider the evolution of a set of probabilities. 
For further simplicity we now take the case $\eta=\delta$ (we could always change basis to achieve this).
To discuss standard probabilities, we need to introduce some choice of a polyhedral convex cone
containing the convex cone of states (see~\cite{Heller} for a related construction).
For that recall the 
indicator functions of Example~\ref{conifer}
\begin{align*}
X_1&=\tfrac14+\tfrac{\sqrt 8 \hh \theta^1+ i \theta^1\theta^2}{12}\, ,
\qquad\qquad\quad \:
X_2=\tfrac14
+\tfrac{-\sqrt2\theta^1+\sqrt 6 \theta^2
+ i \theta^1\theta^2}{12}\, ,\\
X_3&=\tfrac14
+\tfrac{-\sqrt2\theta^1-\sqrt 6 \theta^2
+ i \theta^1\theta^2}{12}\, ,\qquad
\;\;X_4 =\tfrac14
+\tfrac{-3i\theta_1\theta_2}{12}\, .
\end{align*}
We next write a state function in terms of probabilities corresponding to expectations of these indicators.
Note that, unlike states, state functions can have vanishing body. From Equation~\nn{expect} we see that any state function with vanishing body $\phi=0$ here has indicator expectations
$$
\langle X_i\rangle_{}=\frac14 \, ,\quad i=1,\ldots, 4\, .
$$
Otherwise, for $\phi\neq 0$,  we may first normalize $\Phi$ along a fixed $t=T$ laboratory such that
its star-square has unit body 
\begin{equation}\label{nrlize}
(\Phi,\Phi)_{T} :=
\big(\Phi{\filmstar_{\sss\nabla}}\Phi\big)_0=	\phi^2 + \phi_a \delta^{ab} \phi_b + \varphi^2=1\, ,
\end{equation}
which simplifies Equation~\nn{expect}. Calling $p_i:=\langle X_i\rangle\in [0,1]$ we then have
\begin{align*}
p_1&=\tfrac14+\phi \left(\tfrac{\sqrt 8 \hh \phi_1-\varphi}{6}\right)\, ,
\qquad\qquad\quad \: 
p_2=\tfrac14
+\phi \left(\tfrac{-\sqrt2\phi_1+\sqrt 6 \phi_2
-\varphi}{6}\right)\, ,\\
p_3&=\tfrac14
+\phi \left(\tfrac{-\sqrt2\phi_1-\sqrt 6 \phi_2
-\varphi}{6}\right)\, ,\qquad
\;\;p_4 =\tfrac14
+\tfrac{\phi\varphi}{2}\, .
\end{align*}
Since $\phi\neq 0$, we may solve $(\phi_1,\phi_2,\varphi)$ in terms of $(\phi,p_1,p_2,p_3)$

$$
	\phi_1=\frac{2p_1-p_2-p_3}{\sqrt{2}\phi},\quad \phi_2=\frac{\sqrt{3}(p_2-p_3)}{\sqrt{2}\phi},\quad \varphi=\frac{3-4(p_1+p_2+p_3)}{2\phi}\, .
$$
Finally we can find $\phi$ by plugging the previous list of expressions for $\phi_a,\varphi$ in \eqref{nrlize}. While the resulting quartic equation itself has four solutions, we choose the root
$$
\phi=\frac{\sqrt{1+2\sqrt{-\Delta_{\rm ellipsoid}
}}}{\sqrt{2}}\, ,$$
where $\frac{\Delta_{\rm ellipsoid}}{6}:=(p_1+p_2+p_3)^2-(p_1+p_2+p_3+p_1p_2+p_2p_3+p_1p_3)+\frac{1}{3}$.
This makes sense only when $\Delta_{\rm ellipsoid}\leq 0$, implying that the supersymplectic probabilities are constrained to lie inside the ellipsoid of  Equation~\eqref{constraint}in $\mathbb{R}_{\geq 0}^3$. 

\newcommand{\cccc}{{\frak c}}
\newcommand{\ssss}{{\frak s}}

It remains to study the evolution of probabilities. Returning to Equation~\nn{letsmove} and  calling 
$\cccc(t) :=\cos\big(\frac{(t-T)H}2\big)$, 
$\ssss(t) :=\sin\big(\frac{(t-T)H}2\big)$, we obtain
$$
\phi_1(t) =\hh\hh \hh\:\:\cccc(t) \phi_1(T)+\hh
 \ssss(t) \phi_2(T)\, ,
 $$
 $$
\phi_2(t)  =-\ssss(t) \phi_1(T)+
 \cccc(t) \phi_2(T)\, .
$$
Thus
\begin{align*}
	\begin{pmatrix}
		p_1(t)\\[1.5mm]p_2(t)\\[1.5mm] p_3(t)\\[1.5mm] p_4(t)
	\end{pmatrix}=\begin{pmatrix}
		\frac{1+2\cccc(t)}{3}&\frac{1-\cccc(t)+\sqrt{3}\ssss(t)}{3}&\frac{1-\cccc(t)-\sqrt{3}\ssss(t)}{3} &0\\
		\frac{1-\cccc(t)-\sqrt{3}\ssss(t)}{3}& \frac{1+2\cccc(t)}{3}&\frac{1-\cccc(t)+\sqrt{3}\ssss(t)}{3}& 0\\
		\frac{1-\cccc(t)+\sqrt{3}\ssss(t)}{3}& \frac{1-\cccc(t)-\sqrt{3}\ssss(t)}{3}&\frac{1+2\cccc(t)}{3}&0\\[1mm]
		0&0&0&1
	\end{pmatrix}\begin{pmatrix}
		p_1(T)\\[1.5mm]p_2(T)\\[1.5mm] p_3(T)\\[1.5mm]p_4(T)
	\end{pmatrix}\, .
	\end{align*}
	\begin{figure}
	\begin{tikzpicture}
\pgftext{
	\includegraphics[scale=.32]{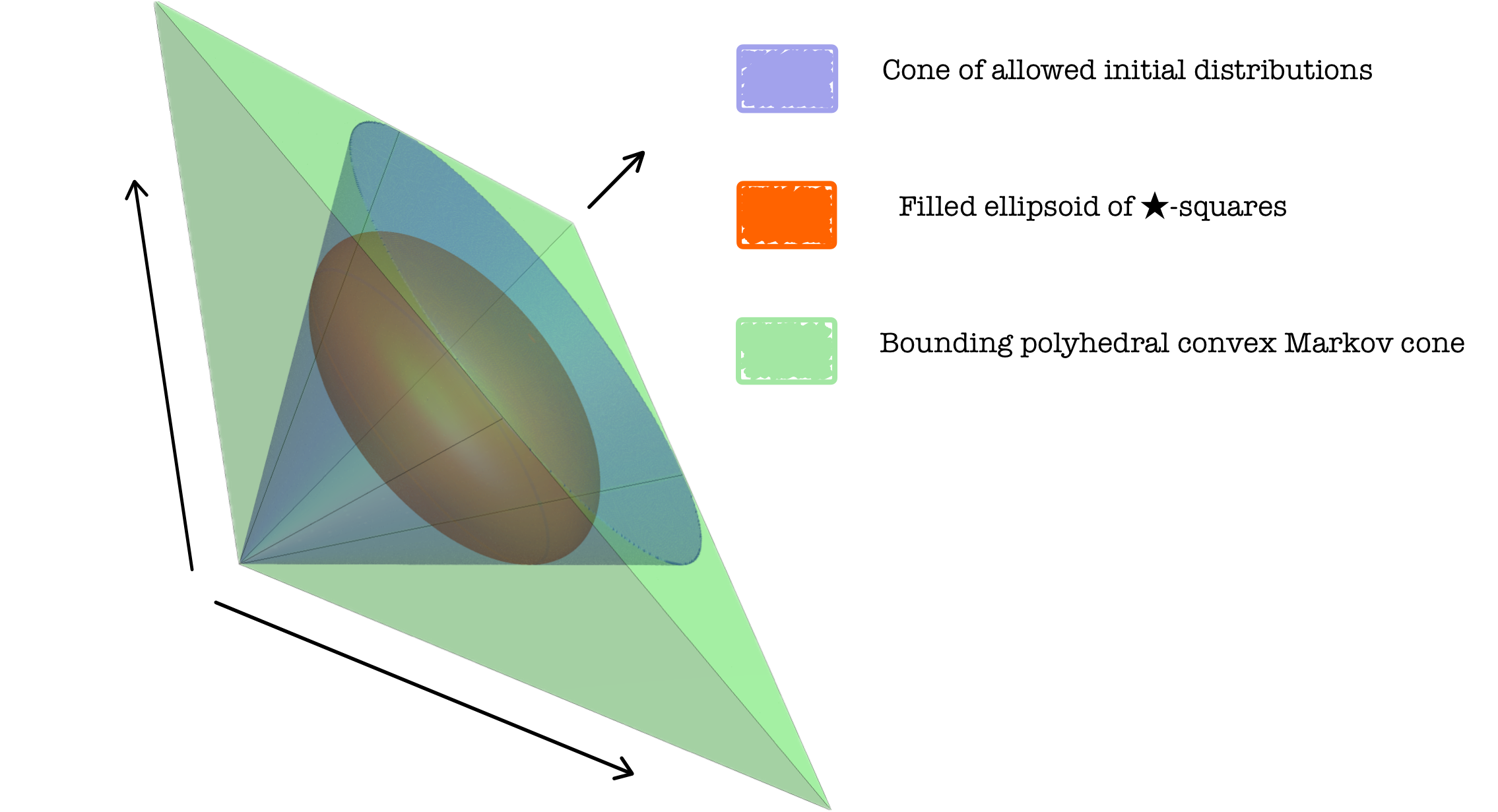}}

	\node[] at (3,2.5) {\scriptsize  $\Delta_{\rm cone}\leq 0$};
	\node[] at (3,1.3) {\scriptsize $\Delta_{\rm ellipsoid}\leq 0$};
	\node[] at (3,0.1) {\scriptsize \quad \quad \quad \quad $p_1,p_2,p_3\geq 0, \: \: p_1+p_2+p_3\leq 1$};
	
	\node[] at (-3,-2.7) {  $p_1$};
	\node[] at (-5.3,.2) {  $p_2$};
	\node[] at (-1.4,2.1) {  $p_3$};

\end{tikzpicture}	\caption{Dynamics corresponds to a path inside the orange ellipsoid}
	\label{fig:ellipsoid}
\end{figure}
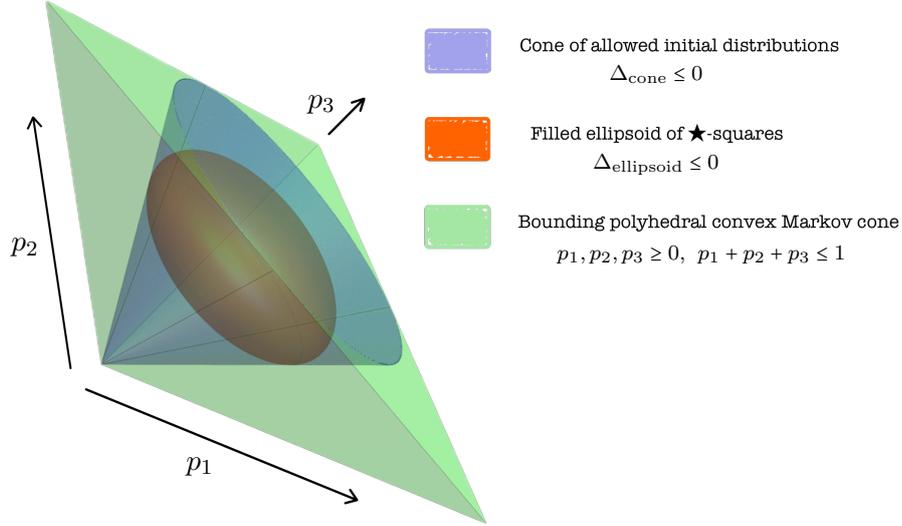

	For probabilities to evolve to probabilities, examining the above transition matrix we see that  the initial probabilities must obey
the conical constraint	
	$$\Delta_{\rm cone}:= p_1^2+p_2^2+p_3^2-2(p_1p_2+p_1p_3+p_2p_3)\leq 0\, .$$
	This is automatically satisfied for supersymplectic probability distributions, {\it videlicet} those that satisfy $\Delta_{\rm ellipsoid}\leq 0$ or lie inside the ellipsoid of  Equation~\eqref{constraint}, as depicted in Figure~\ref{fig:ellipsoid}. 
While the above transition matrix is not Markov  (its  entries in any row or column sum to unity without being individually strictly non-negative)  its 
eigenvalues do have magnitudes  bounded by unity. It thus defines a Markov-like stochastic process.

\section*{Acknowledgements}
S.C. acknowledges a U.C. Davis Dean's Summer Fellowship.
A.W. acknowledges support from the Royal Society of New Zealand via Marsden Grant 19-UOA-008.  A.W. was also supported by  the Simons Foundation Collaboration Grant for Mathematicians~ID 686131. 
C.Y. acknowledges support of the Scientific and Technological Research Council of Turkey (TUBITAK) 2214-A Grant No. 1059B141900615.

%
%
	

%

\newpage


\begin{thebibliography}{10}

\bibitem{brink1976local}
L.~Brink, S.~Deser, B.~Zumino, P.~Di~Vecchia, and P.~Howe.
\newblock Local supersymmetry for spinning particles.
\newblock {\em Physics Letters B}, 64(4):435--438, 1976.

\bibitem{casalbuoni1976classical}
R.~Casalbuoni.
\newblock The classical mechanics for bose-fermi systems.
\newblock {\em Nuovo Cimento A}, 33:389--431, 1976.

\bibitem{Berezin:1976eg}
F.~A. Berezin and M.~S. Marinov.
\newblock {Particle Spin Dynamics as the Grassmann Variant of Classical
  Mechanics}.
\newblock {\em Annals Phys.}, 104:336, 1977.

\bibitem{Galvao:1980cu}
C.~A.~P. Galvao and C.~Teitelboim.
\newblock {Classical Supersymmetric Particles}.
\newblock {\em J. Math. Phys.}, 21:1863, 1980.

\bibitem{bastianelli2006path}
F.~Basianelli and P.~van Nieuwenhuizen.
\newblock {\em Path integrals and anomalies in curved space}.
\newblock Cambridge University Press, 2006.

\bibitem{Barducci}
A.~Barducci, R.~Casalbuoni, and L.~Lusanna.
\newblock A possible interpretation of theories involving grassmann variables.
\newblock {\em Lett. Nuovo Cim.}, 19:581, 1977.

\bibitem{koopman1931hamiltonian}
B.~O. Koopman.
\newblock Hamiltonian systems and transformation in {H}ilbert space.
\newblock {\em Proceedings of the National Academy of Sciences},
  17(5):315--318, 1931.

\bibitem{neumann1932operatorenmethode}
J.~von Neumann.
\newblock Zur operatorenmethode in der klassischen mechanik.
\newblock {\em Annals of Mathematics}, pages 587--642, 1932.

\bibitem{MR0229429}
G.~Ludwig.
\newblock Attempt of an axiomatic foundation of quantum mechanics and more
  general theories. {II}.
\newblock {\em Comm. Math. Phys.}, 4:331--348, 1967.

\bibitem{MR0229430}
G.~Ludwig.
\newblock Attempt of an axiomatic foundation of quantum mechanics and more
  general theories. {III}.
\newblock {\em Comm. Math. Phys.}, 9:1--12, 1968.

\bibitem{MR0249064}
G.~D\"{a}hn.
\newblock Attempt of an axiomatic foundation of quantum mechanics and more
  general theories. {IV}.
\newblock {\em Comm. Math. Phys.}, 9:192--211, 1968.

\bibitem{MR0242429}
P.~Stolz.
\newblock Attempt of an axiomatic foundation of quantum mechanics and more
  general theories. {V}.
\newblock {\em Comm. Math. Phys.}, 11:303--313, 1968/69.

\bibitem{MR0297221}
P.~Stolz.
\newblock Attempt of an axiomatic foundation of quantum mechanics and more
  general theories. {VI}.
\newblock {\em Comm. Math. Phys.}, 23:117--126, 1971.

\bibitem{MR0342092}
S.~Gudder.
\newblock Convex structures and operational quantum mechanics.
\newblock {\em Comm. Math. Phys.}, 29:249--264, 1973.

\bibitem{schwarz2020geometric}
A.~Schwarz.
\newblock Geometric approach to quantum theory.
\newblock {\em SIGMA. Symmetry, Integrability and Geometry: Methods and
  Applications}, 16:020, 2020.

\bibitem{MR3752196}
I.~Bengtsson and K.~\.{Z}yczkowski.
\newblock {\em Geometry of quantum states}.
\newblock Cambridge University Press, Cambridge, second edition, 2017.
\newblock An introduction to quantum entanglement.

\bibitem{MR2492178}
L.~D. Faddeev and O.~A. Yakubovskii.
\newblock {\em Lectures on quantum mechanics for mathematics students},
  volume~47 of {\em Student Mathematical Library}.
\newblock American Mathematical Society, Providence, RI, 2009.
\newblock Translated from the 1980 Russian original by Harold McFaden, With an
  appendix by Leon Takhtajan.

\bibitem{MR2433906}
L.~A. Takhtajan.
\newblock {\em Quantum mechanics for mathematicians}, volume~95 of {\em
  Graduate Studies in Mathematics}.
\newblock American Mathematical Society, Providence, RI, 2008.

\bibitem{MR2222127}
F.~Strocchi.
\newblock {\em An introduction to the mathematical structure of quantum
  mechanics}, volume~27 of {\em Advanced Series in Mathematical Physics}.
\newblock World Scientific Publishing Co. Pte. Ltd., Hackensack, NJ, 2005.
\newblock A short course for mathematicians.

\bibitem{arnol2013mathematical}
V.~I. Arnol'd.
\newblock {\em Mathematical methods of classical mechanics}, volume~60.
\newblock Springer Science \& Business Media, 2013.

\bibitem{ArnGiv}
V.~I. Arnol'd and A.~B. Givental.
\newblock Symplectic geometry [{MR}0842908 (88b:58044)].
\newblock In {\em Dynamical systems, {IV}}, volume~4 of {\em Encyclopaedia
  Math. Sci.}, pages 1--138. Springer, Berlin, 2001.

\bibitem{busch2016quantum}
P.~Busch, P.~Lahti, J-P Pellonp{\"a}{\"a}, and K.~Ylinen.
\newblock {\em Quantum measurement}, volume~23.
\newblock Springer, 2016.

\bibitem{herczeg2018contact}
G.~Herczeg and A.~Waldron.
\newblock Contact geometry and quantum mechanics.
\newblock {\em Physics Letters B}, 781:312--315, 2018.

\bibitem{he2010odd}
Z.~He et~al.
\newblock {\em Odd dimensional symplectic manifolds by Zhenqi He.}
\newblock PhD thesis, Massachusetts Institute of Technology, 2010.

\bibitem{lin2013lefschetz}
Y.~Lin.
\newblock Lefschetz contact manifolds and odd dimensional symplectic geometry.
\newblock {\em arXiv preprint arXiv:1311.1431}, 2013.

\bibitem{cappelletti2013survey}
B.~Cappelletti-Montano, A.~De~Nicola, and I.~Yudin.
\newblock A survey on cosymplectic geometry.
\newblock {\em Reviews in Mathematical Physics}, 25(10):1343002, 2013.

\bibitem{manin1997introduction}
Y.~I. Manin.
\newblock Introduction to supergeometry.
\newblock {\em Gauge Field Theory and Complex Geometry}, pages 181--232, 1997.

\bibitem{witten2012notes}
E.~Witten.
\newblock Notes on supermanifolds and integration.
\newblock {\em arXiv preprint arXiv:1209.2199}, 2012.

\bibitem{Deligne}
P.~Deligne, P.~Etingof, D.~S. Freed, L.~C. Jeffrey, D.~Kazhdan, J.~W. Morgan,
  D.~R. Morrison, and E.~Witten, editors.
\newblock {\em {Quantum fields and strings: A course for mathematicians. Vol.
  1, 2}}.
\newblock American Mathematical Society, 1999.

\bibitem{batchelor1979structure}
M.~Batchelor.
\newblock The structure of supermanifolds.
\newblock {\em Transactions of the American Mathematical Society},
  253:329--338, 1979.

\bibitem{koszul1994connections}
J-L Koszul.
\newblock Connections and splittings of supermanifolds.
\newblock {\em Differential Geometry and its Applications}, 4(2):151--161,
  1994.

\bibitem{Schwarz1989}
A.~S. Schwarz.
\newblock Symplectic, contact and superconformal geometry, membranes and
  strings.
\newblock 1 1990.

\bibitem{Rothstein1990}
M.~Rothstein.
\newblock {The Structure of supersymplectic supermanifolds}.
\newblock In {\em {19th International Conference on Differential Geometrical
  Methods in Theoretical Physics}}, pages 331--343, 3 1990.

\bibitem{Schwarz1993}
A.~S. Schwarz.
\newblock {Geometry of Batalin-Vilkovisky quantization}.
\newblock {\em Commun. Math. Phys.}, 155:249--260, 1993.

\bibitem{schwarz1994superanalogs}
A.~Schwarz.
\newblock Superanalogs of symplectic and contact geometry and their
  applications to quantum field theory.
\newblock {\em arXiv preprint hep-th/9406120}, 1994.

\bibitem{curtright2013concise}
T.~L. Curtright, D.~B. Fairlie, and C.~K. Zachos.
\newblock {\em A concise treatise on quantum mechanics in phase space}.
\newblock World Scientific Publishing Company, 2013.

\bibitem{woit2017quantum}
P.~Woit.
\newblock {\em Quantum theory, groups and representations}.
\newblock Springer, 2017.

\bibitem{Corradini}
O.~Corradini, E.~Latini, and A.~Waldron.
\newblock Quantum {D}arboux theorem.
\newblock {\em Phys. Rev. D}, 103(10):Paper No. 105021, 19, 2021.

\bibitem{rogers2007supermanifolds}
A.~Rogers.
\newblock {\em Supermanifolds: theory and applications}.
\newblock World Scientific, 2007.

\bibitem{vaquie2021sheaves}
M.~Vaquie.
\newblock Sheaves and functors of points.
\newblock {\em New Spaces in Mathematics: Volume 1: Formal and Conceptual
  Reflections}, page 407, 2021.

\bibitem{monterde1993existence}
J.~Monterde and O.~A. S{\'a}nchez-Valenzuela.
\newblock Existence and uniqueness of solutions to superdifferential equations.
\newblock {\em Journal of Geometry and Physics}, 10(4):315--343, 1993.

\bibitem{garnier2013integration}
S.~Garnier and T.~Wurzbacher.
\newblock Integration of vector fields on smooth and holomorphic
  supermanifolds.
\newblock {\em Documenta Mathematica}, 18:519--545, 2013.

\bibitem{jetzer1999completely}
F.~Jetzer.
\newblock {\em Completely integrable systems on supermanifolds}.
\newblock PhD thesis, McGill University, 1999.

\bibitem{fuchscohomology}
D.~B. Fuchs.
\newblock Cohomology of infinite-dimensional lie algebras.
\newblock {\em New York, Consultants Bureau}, 1986.

\bibitem{bastianelli2009detours}
F.~Bastianelli, O.~Corradini, and A.~Waldron.
\newblock Detours and paths: Brst complexes and worldline formalism.
\newblock {\em Journal of High Energy Physics}, 2009(05):017, 2009.

\bibitem{Berezin}
F.~A. Berezin.
\newblock {\em Introduction to superanalysis}, volume~9.
\newblock Springer Science \& Business Media, 2013.

\bibitem{schmitt1990supergeometry}
T.~Schmitt.
\newblock Supergeometry and hermitian conjugation.
\newblock {\em Journal of Geometry and Physics}, 7(2):141--169, 1990.

\bibitem{rothstein1991structure}
M.~Rothstein.
\newblock The structure of supersymplectic supermanifolds.
\newblock In {\em Differential Geometric Methods in Theoretical Physics}, pages
  331--343. Springer, 1991.

\bibitem{sasaki1958differential}
S.~Sasaki.
\newblock On the differential geometry of tangent bundles of riewannian
  manifolds.
\newblock {\em Tohoku Mathematical Journal, Second Series}, 10(3):338--354,
  1958.

\bibitem{Cast}
L.~Castellani, R.~Catenacci, and P.A. Grassi.
\newblock Hodge dualities on supermanifolds.
\newblock {\em Nuclear Physics B}, 899:570--593, 2015.

\bibitem{HT}
Marc Henneaux and Claudio Teitelboim.
\newblock {\em Quantization of Gauge Systems}.
\newblock Princeton University Press, 1992.

\bibitem{fedosov1994simple}
B.~V. Fedosov.
\newblock A simple geometrical construction of deformation quantization.
\newblock {\em Journal of differential geometry}, 40(2):213--238, 1994.

\bibitem{berezin1980feynman}
F.~A. Berezin.
\newblock Feynman path integrals in a phase space.
\newblock {\em Soviet Physics Uspekhi}, 23(11):763, 1980.

\bibitem{fairlie1989infinite}
D.~B. Fairlie and C.~K. Zachos.
\newblock Infinite-dimensional algebras, sine brackets, and su($\infty$).
\newblock {\em Physics Letters B}, 224(1-2):101--107, 1989.

\bibitem{fradkin1991quantization}
E.~S. Fradkin and V.~Ya Linetsky.
\newblock Quantization and cocycles on the supertorus and large-n limits for
  the classical lie superalgebras.
\newblock {\em Modern Physics Letters A}, 6(03):217--224, 1991.

\bibitem{hirshfeld2002deformation}
A.~C. Hirshfeld and P.~Henselder.
\newblock Deformation quantization for systems with fermions.
\newblock {\em Annals of Physics}, 302(1):59--77, 2002.

\bibitem{hirshfeld2004cliffordization}
A.~C. Hirshfeld, P.~Henselder, and T.~Spernat.
\newblock Cliffordization, spin, and fermionic star products.
\newblock {\em Annals of Physics}, 314(1):75--98, 2004.

\bibitem{galaviz2008weyl}
I.~Galaviz, H.~Garcia-Compean, M.~Przanowski, and F.~J. Turrubiates.
\newblock Weyl--wigner--moyal formalism for fermi classical systems.
\newblock {\em Annals of Physics}, 323(2):267--290, 2008.

\bibitem{Bor}
Martin Bordemann.
\newblock The deformation quantization of certain super-poisson brackets and
  brst cohomology.
\newblock In {\em Conf{\'e}rence Mosh{\'e} Flato 1999: Quantization,
  Deformations, and Symmetries Volume II}, pages 45--68. Springer, 2000.

\bibitem{Hir1}
A.~C. Hirshfeld and P.~Henselder.
\newblock Deformation quantization for systems with fermions.
\newblock {\em Annals of Physics}, 302(1):59--77, 2002.

\bibitem{Hir2}
A.~C. Hirshfeld, P.~Henselder, and T.~Spernat.
\newblock Cliffordization, spin, and fermionic star products.
\newblock {\em Annals of Physics}, 314(1):75--98, 2004.

\bibitem{VP}
A.~Van~Proeyen.
\newblock {Tools for supersymmetry}.
\newblock {\em Ann. U. Craiova Phys.}, 9(I):1--48, 1999.

\bibitem{Schneider}
R.~Schneider.
\newblock {\em Convex cones---geometry and probability}, volume 2319 of {\em
  Lecture Notes in Mathematics}.
\newblock Springer, Cham, 2022.

\bibitem{bayen1978deformation}
F.~Bayen, M.~Flato, C.~Fronsdal, A.~Lichnerowicz, and D.~Sternheimer.
\newblock Deformation theory and quantization. i. deformations of symplectic
  structures.
\newblock {\em Annals of Physics}, 111(1):61--110, 1978.

\bibitem{lichnerowicz1982deformations}
A.~Lichnerowicz.
\newblock D{\'e}formations d’alg{\`e}bres associ{\'e}es {\`a} une
  vari{\'e}t{\'e} symplectique (les $\filledstar$-produits).
\newblock In {\em Annales de l'institut Fourier}, volume 32(1), pages 157--209,
  1982.

\bibitem{Heller}
A.~Heller.
\newblock On stochastic processes derived from {M}arkov chains.
\newblock {\em Ann. Math. Statist.}, 36:1286--1291, 1965.

\end{thebibliography}
\end{document}